\RequirePackage{fix-cm} 
\documentclass[smallextended]{svjour3}       
\smartqed  
\usepackage{graphicx}
\usepackage{subfigure}
\usepackage[raggedright]{sidecap}
\usepackage{mathptmx}
\usepackage{amsmath}
\usepackage{amssymb}
\usepackage[utf8]{inputenc}
\usepackage[english]{babel}
\usepackage[authoryear, round]{natbib}
\usepackage{mathtools}
\DeclarePairedDelimiter\ceil{\lceil}{\rceil}
\DeclarePairedDelimiter\floor{\lfloor}{\rfloor}
\usepackage{footmisc}
\usepackage[normalem]{ulem}

\usepackage{color}
\usepackage[usenames,dvipsnames,svgnames,table]{xcolor}

\usepackage{hyperref}
\hypersetup{
     colorlinks   = true,
     citecolor    = blue
}

\newcounter{clm1}
\newtheorem{clm}[clm1]{Claim}
\newcommand\myworries[1]{\textcolor{black}{#1}}
\definecolor{mycolor}{rgb}{0, 0, 0}
\usepackage{footmisc}
\usepackage{latexsym}
\begin{document}

\title{Stability, Efficiency, and Contentedness of Social Storage Networks}

\titlerunning{Stability of Social Storage Networks}        

\author{Pramod Mane \and \\
        Kapil Ahuja \and \\
        Nagarajan Krishnamurthy
}

\authorrunning{P. Mane, K. Ahuja, N. Krishnamurthy} 

\institute{P. Mane and K. Ahuja \at
	 Computer Science and Engineering, \\ 
	Indian Institute of Technology Indore, India. \\ 
              \email{pramod@iiti.ac.in, kahuja@iiti.ac.in}           
		\and
           N. Krishnamurthy \at
	Operations Management and Quantitative Techniques, \\ 
	Indian Institute of Management Indore, India.\\ 
	\email{nagarajan@iimidr.ac.in, naga.research@gmail.com}
}

\maketitle

\begin{abstract}
Social storage systems are a good alternative to existing data backup systems of local, centralized, and P2P backup. Till date, researchers have mostly focussed on either building such systems by using existing underlying social networks (exogenously built) or on studying Quality of Service (QoS) related issues. In this paper, we look at two untouched aspects of social storage systems. One aspect involves modelling social storage as an endogenous social network, where agents themselves decide with whom they want to build data backup relation, which is more intuitive than exogenous social networks.  The second aspect involves studying the stability of social storage systems, which would help reduce maintenance costs and further, help build efficient as well as  contented networks.

\myworries{We have a four fold contribution that covers the above two aspects. We, first, model the social storage system as a {\it strategic network formation game}. We define the utility of each agent in the network under two different frameworks, one where the cost to add and maintain links is considered in the utility function and the other where budget constraints are considered. In the context of social storage and social cloud computing, these utility functions are the first of its kind, and we use them to define and analyse the social storage network game. Second, we propose the concept of {\it bilateral stability} which refines the pairwise stability concept defined by \cite{jackson}, by requiring mutual consent for both addition and deletion of links, as compared to mutual consent just for link addition. Mutual consent for link deletion is especially important in the social storage setting. The notion of {\it bilateral stability} subsumes the bilateral equilibrium definition of \cite{Goyal-bilateral-deviation-important-paper}.  Third, we prove necessary and the sufficient conditions for bilateral stability of social storage networks. For symmetric social storage networks, we prove that there exists a unique neighborhood size, independent of the number of agents (for all non-trivial cases), where no pair of agents has any incentive to increase or decrease their neighborhood size. We call this neighborhood size as the \textit{stability point}. Fourth, given the number of agents and other parameters, we discuss which bilaterally stable networks would evolve and also discuss which of these stable networks are efficient --- that is, stable networks with maximum sum of utilities of all agents. We also discuss ways to build contented networks, where each agent achieves the maximum possible utility. }

\keywords{Social Storage, Endogenous Network Formation, Bilateral Stability, Pairwise Stability, F2F Backup System, Peer-to-Peer System}

\end{abstract}

\section{Introduction}\label{sec:introduction}
In this digital era, where personal data size is growing exponentially, data backup is not a new need. Data stored on an agent's local machine is prone to loss due to disk-failure, malware, and so on. Local backup, \textcolor{mycolor}{centralised on-line backup (for example, Backblaze\footnote{\label{footnote:backblaze}https://secure.backblaze.com/buy.htm (Visited on 09 May 2017)}, CrashPlan\footnote{\label{footnote:crash-plan}https://store.crashplan.com/store/ (Visited on 09 May 2017)}) and decentralised (Peer-to-Peer) backup (for example, Pstore \citep{pstore}, Pastiche \citep{pastiche}, Samsara \citep{samsara}, etc.)} are some strategies available to agents. Each has its own merits and demerits. For example, maintaining data backup on a local external hard disk on a regular basis is cumbersome. As far as on-line backup systems are concerned, on the one hand, centralised on-line backup is not cost efficient, especially when the amount of data required to be backed up is huge. On the other hand, although Peer-to-Peer (P2P) backup systems are cost efficient, they require dealing with several issues like data \textcolor{mycolor}{availability, reliability and security \citep{Steinmetz2005}}. 

In recent years, to cope up with the above issues in P2P storage systems, researchers have been focusing on social network relationships. It is believed that social ties between agents will help to build backup systems that overcome aforementioned issues. This trend that takes real world social relationships into account for constructing a data backup system is emerging as \textcolor{mycolor}{{\it Social Storage} or Friend-to-Friend (F2F) Storage (Friendstore \citep{friendstore-1}, FriendBox \citep{friendboxstorage}, BackupBuddy\footnote{http://www.buddybackup.com/ (Visited on 09 May 2017)\label{footnote:buddybackup}} are a few examples)}.

\textcolor{mycolor}{Existing research on social storage is moving in two directions. One research trend \citep{f2fstorage, blockparty, friendstore, f2box, hybridfriendbox, friendboxstorage} has been focusing on various technical approaches to build the system, for example, data backup techniques. 
The other direction \citep{f2frecoverability, f2femepricalstudy, Gracia-availability-f2f, f2findirecttie} has been focusing on studying Quality of Service (QoS) related issues. This includes data availability, reliability, the cost associated with communication, data maintenance, data placement or scheduling polices, by taking online social relationships into account.} 

\subsection{Social Storage Issues Addressed in this Paper}
In this work, we address two other issues of social storage. {\it First}, major social storage studies (such as \cite{Gracia-availability-f2f, hybridfriendbox, f2findirecttie}) have considered exogenous social networks (an underlying social network, for instance Facebook, Orkut, Venus, etc.) to construct a social storage system and to study QoS related issues\footnote{This is because social storage is in its infancy and an architectural prototype of social storage is in the development stage.}. However, the approach of considering an exogenous social network to build a social storage system (or to do QoS analysis) fails to address various aspects. 

Here, there is an assumption that an agent in the underlying network is involved in data backup activity with all its neighbors. However, it is possible that agents do not want to perform a data backup activity with their set of existing neighbors. In other words, the approach does not focus on participation benefits and costs. Rational (self-interested) behavior of agents involved in the data backup activity is not taken into consideration\footnote{Although \cite{f2femepricalstudy} begin discussing about agents' strategic behavior in a scenario where limited storage is available for the agents, this has just been touched upon and has not been looked at in detail.}. 

The QoS analysis, which is based upon the neighborhood size in the underlying network, is no longer valid. Thus, it is important to study when agents want to perform a data backup activity and/ or when they do not.
Hence, in this paper, we model the social storage system as an {\it endogenous network formation game}. 


{\it Second}, social storage systems may not be stable (when agents have no incentive to add new partners or delete existing partners). Even if stable, they may not be efficient (maximizing the sum of utilities of all agents), and even if efficient,  they may not be contended (when all agents achieve their maximum utility). 

There is limited study on stability, efficiency, and contentment of social storage systems. While proposing the idea of F2F backup systems, \cite{f2fstorage} argue that social ties between agents act as incentives for them to stay in the system, thereby resulting in a stable social storage system. In their context, a system is unstable when agents arrive and depart the social storage system randomly --- lesser this randomness, more the stability of the system. In our case, a social storage system, as above, is {\it stable} when agents have no incentive to add new partners or delete existing partners. In the following subsections, we motivate this definition of stability in detail. 

Studying social storage systems as an endogenous network formation game, and then analyzing its stability may, on first glance, seem contradictory --- since one cannot do anything from outside the (endogenous) system if the agents' themselves do not form a stable network. 
In our case, agents always form a stable network, but the network may not be {\it efficient}. That is, the sum of the utilities of all agents may not be the maximum possible. In our case, as many agents as possible may not be {\it contented} as well, i.e. all agents achieve maximum possible utility. Contented networks are also efficient. 

Looking at both endogenous network formation and efficiency and contentment is useful because, though the social storage system is built endogenously, an independent observer (say, an administrator or a regulator) can check whether the system is efficient and contented, and if not, can externally do a small perturbation to the network. In some scenarios we look at, the independent administrator may achieve efficiency and/ or contentment by just introducing a small number of dummy agents.

\subsection{Our Model}
Both aspects of social storage systems that we address in this paper (that is, endogenously evolved systems as well as stability, efficiency  and contentment of such systems) are easily analysed by using {\it strategic network formation} models. \myworries{Modelling a utility function (the payoff that each agent receives in a network) is the foremost requirement to study network formation in a strategic setting \citep{Jackson-Book}.  This aspect has not been given much attention by researchers working in the social storage domain.}


\myworries{In the strategic network formation literature, specifically endogenous network formation game, different kinds of utility functions have been proposed and successfully validated. We summarise some of these utility functions in the next section. Utility modelling is more crucial in the social storage context, where decision makers are human agents who aim to optimise their own goals. This is in comparison to the P2P storage context (our closest cousin), where nodes (computer systems) are decision makers. Here, agents do not want to loose their data and want to maximise their data reliability, which is also conceived as a risk averse behaviour.}

\myworries{In this paper, we compute the utility of agents by incorporating some fundamental aspects such as the disk failure rate, the value (or benefit) associated with the data, and the cost to an agent for maintaining relationship with others. We discuss this utility function under two different frameworks, namely Multi-Objective Framework and Single-Objective Framework. As far as we know, in the social storage literature, this is the first attempt of its kind. The most challenging aspect of designing the utility function used by us is that it is simple, yet captures the behavior of the system well. Our utility model can also easily extendible to more realistic scenarios involving the online availability of agents, the bandwidth available to them, the agent heterogeneity, and trust.}

\subsection{Our Solution Concept and Technique}
In the endogenous network formation model of \cite{jackson}, rational decision makers build a network by interacting with each other. Here, the pairwise stability solution concept takes agents' mutual consent into account while building a relationship (that is, adding a link in the network). But, any agent can decide not to maintain a relationship (that is, delete any of its existing links) without consent of the agent at the other end of the link. 

However, the social storage system discussed earlier, impels us to focus on the requirement of bilateral consent while deleting a link as well. For instance, let agents $i$ and $j$ be backup partners. That is, $i$ provides its storage space to $j$ for the purpose of storing $j's$ data, and vice versa. Now, let us assume that breaking a backup partnership without mutual consent is allowed. If agent $i$ breaks the partnership without consent of $j$, then there is a threat that $j$ will lose its data which is stored on $i's$ storage space. Hence, backup partnerships in social storage networks have to be viewed as mutual contracts which cannot be broken unilaterally. We call this as {\it bilateral stability}. This definition of bilateral stability also applies to other contexts where mutual consent is required for deletion, for example, service level agreements in the Cloud. \myworries{We give a brief overview of the strategic network formation literature (specifically, link formation strategies), other solution concepts, and how bilateral stability relates to them in the next section.}

\myworries{The rest of this paper is divided into seven more sections. In Section \ref{sec:the-background}, we discuss past work done in strategic network formation. Here, we glance at different utility functions and different solution concepts. In Section \ref{sec:the-model}, we formally describe our social storage model and compute the utility of agents. In Section \ref{sec:background}, we study the endogenous social storage network formation game by focusing on mutual consent for both link addition and deletion. That is, we propose our solution concept of bilateral stability. In Section \ref{sec:stability-point}, we provide some necessary and sufficient conditions for bilateral stability of social storage networks. Here, we first define the {\it stability point} (the ideal neighborhood size) such that no agent gains by deviating from the stability point. Then, we show that there exists a unique stability point independent of the number of agents (for all non-trivial cases).  We also show that there exist unique and non-unique pairwise stable storage networks under certain conditions. Also, given the number of agents and other parameters, we discuss which pairwise stable networks would evolve. In Section \ref{subsec:stabile-networks}, we do further analysis related to stability, efficiency as well as contentment of networks, e.g. efficient networks are always stable in the context we study. In Section \ref{sec:related-future}, we look at related work, and in Section \ref{sec:conclusion} we conclude the paper and discuss future work.}

\section{Background}\label{sec:the-background}
\myworries{Strategic network formation literature is vast. In rest of this section, we first summarise few seminal works in utility function design. Then, we visit few important works in strategic network formation modelling.  Finally, we end this section with a brief survey of impactful solution concepts. We also relate our contribution to the state-of-the-art work in each of these three areas.}

\myworries{We do not focus on papers which look at applications \citep{belleflamme2004market, goyal2006bilateralism, furusawa2007free, venkatesh2013strategic, goyal2001r, goyal2003networks, zirulia2006industry, Suijs-Network-2005, Skorin-Kapov2017} or touch on other tangential topics. For example, anti-coordination among agents \citep{bramoull}; contextual and correlated peer effects \citep{bramoull-peer-effect}; partner heterogeneity \citep{billand-strict-nash-networks}; and Nash and stable characterisation for a graph structure \citep{bramoul-strategic-interaction}.}


\subsection{Utility Function in the Model}
\myworries{As pointed out earlier, a utility function is an important element of a strategic network formation game. This reveals individual benefit-cost tradeoff in a network. Utility functions are either degree based (only direct connections) or distance based (direct as well as indirect connections). }

\myworries{In degree based utility functions, one is only concerned about the effects of its local neighbourhood. This is the case for us too. Although indirect connections are not considered explicitly here, they do effect an agent's utility either positively or negatively. For example, in the co-author model \citep{jackson}, where agents are involved in a collaborative project, an individual's utility goes down if its neighbours are tightly connected (or agents are densely connected in the network). In the job contact network \citep{jobcontact}, for an individual, the probability of getting job information increases as its neighbourhood size increases. However, it also depends on how the individuals are connected (tightly or loosely) and unemployment in the network. }

\myworries{Distance based modelling is suitable for those settings where agents are aiming to minimise the cost of communication. The connection model \citep{jackson}, the network creation game \citep{Fabrikant-NCG-2003}, the locality game \citep{stefan-network-creation-game}, are some examples where a distance based utility function is used. In the social storage context, \cite{f2findirecttie} suggests explicitly exploiting indirect relationships so as to maximise data reliability and availability. In our view, this approach is in-general suitable for those systems that utilise exogenous social relationships (i.e., a social graph). }

\subsection{Network Formation Game in the Model}
\myworries{Most exhaustive survey of network formation games and games on networks has been done in the following works: \cite{Dutta-survey-1, Jackson-survey-1, networkthesis, goyal-on-web, borkotokey-survey}. Few standard models broadly cover strategic network formation modelling. This includes, the cooperative game theory model, the unilateral connection model, the link investments model, and the bilateral connection model (\cite{networkthesis}). Next, we briefly discuss these models and also relate them to our model.} 

\myworries{\cite{aumann1988endogenous} have proposed an extensive network formation game, where agents form links sequentially (one after another) using some exogenous rules. Agents propose with whom they want to form links, and later that proposal is either accepted or rejected by others. But once a link is formed between a pair of agents, it cannot withdrawn. This is the essence of the cooperative game theory model. In the unilateral connection model, agents form links without consent and links are directional (\cite{balanetwork}). 
In the link investments model (\cite{bloch2007formation}) and its variant (\cite{bloch2009communication}), agents propose investments for their every direct link. These investments are either positive or negative. Linking between a pair of agents takes place if and only if total investment on that link is positive.}

\myworries{\cite{myersongraphs} has proposed the link-announcement game. In this agents proposing to form links, announce the name of the agents with whom they want to form these links.  This announcement is done simultaneously. A link between two agents takes place if only if both announce each others name. Inspired by the link-announcement model, \cite{jackson} proposed the pairwise connection model. Here, link formation takes place with the mutual consent of the involved agents, however, link deletion takes place without consent. Our social storage network formation game is inspired by the pairwise connection model. We differ in the link deletion scenario, where for deleting a link, both the agents involved must agree.}

\subsection{Solution Concept}
\myworries{\cite{jackson} have argued that the Nash equilibrium as a solution concept is not useful in the network formation context for two reasons. First, due to the existence of multiple Nash equilibrium, and second, it fails to capture mutual consent of agents in link formation. Hence, they have proposed the pairwise stability solution concept. In this paper, we have shown that there is a need to refine the pairwise stability solution concept in the social storage context. We have changed the deletion condition of the pairwise stability so that it is suitable in our context. We call this concept as bilateral stability.} 

\myworries{Our solution concept of bilateral stability subsumes the concept of bilateral equilibrium proposed by \cite{Goyal-bilateral-deviation-important-paper}. The set of all strategies that are bilaterally stable contains the set of all bilateral equilibrium strategies. A network which is bilaterally stable may contain agents who may be better off by deviating, where as a bilateral equilibrium network does not contain any such agent.  Both definitions, however, allow only bilateral deviations (or pairwise addition as well as deletion with mutual consent). \cite{Buechel-bilateral-stability} have termed bilateral equilibrium as bilateral stability. 
\cite{hummon2000utility} also discusses mutual consent for deletion, but he does not formally define or study the concept of stability with mutual consent for deletion. The author performs agent based simulation of the connection model proposed by \cite{jackson}, and discusses simulation outputs. Other works that focus on agent based simulations are by \cite{falk2003s} and \cite{goeree2009search}.}

\myworries{Other network formation game solution concepts, for example, strong and coalition-proof Nash equilibria \citep{duttamutuswami}; strong pairwise stability \citep{jackson2005strongly}; pairwise stable Nash equilibrium \citep{goyal-pairwise-nash}; farsighted equilibrium \citep{Dutta-Farsighted}; Nash-Cournot equilibrium \citep{Flaam-network-creation-game}; and monadic stability \citep{gilles-formation-under-mutual-consent} are relevant in variants of the scenario we discuss in this paper. }


\section{Social Storage Network Model}\label{sec:the-model}

\begin{definition}
\textit{A social storage network $\mathfrak{g} =(\mathbf{A},\mathbf{L})$ consists of a set of agents, $\mathbf{A}$, and a set of links connecting these agents, $\mathbf{L}$, where a link between two agents represents a data backup partnership between them}.
\end{definition}
\begin{table}[ht!]
\caption{Notation Summary}
\begin{tabular}{ |l|l| }
\hline
$\mathfrak{g}$ & social storage network.\\ \hline
$\mathbf{A}$ & set of agents (or vertices). \\ \hline
$N$ & number of agents in $\mathfrak{g}$ (that is, $N$ is the number of elements in the set $\mathbf{A}$). \\ \hline
$\mathbf{L}$ &set of links (or edges). \\ \hline
$\langle ij\rangle$ & link between agents $i$ and $j$.\\ \hline
$a_{ij}$ & indicator for data backup partnership between agents $i$ and $j$. \\ \hline
$c$ & cost incurred by an agent to maintain a link. \\ \hline
$\beta_{i}$ & worth (or value) that agent $i$ has for its data.\\ \hline 
$s_{i}$ & amount of storage available with agent $i$ that it can contribute to other agents. \\ \hline
$d_{i}$ & amount of data that agent $i$ wants to backup. \\ \hline
$b_{i}$ & budget allocated by agent $i$ towards backup partnerships. \\ \hline
$\lambda$  &  probability of failure of a disk. \\ \hline
$\eta_{i}(\mathfrak{g})$ & neighborhood size of agent $i$ in $\mathfrak{g}$. (Also denotes the set of neighbors of $i$). \\ \hline
$\mathfrak{g}+\langle ij\rangle$ & new link $\langle ij\rangle$ is added to $\mathfrak{g}$. \\ \hline
$\mathfrak{g}-\langle ij\rangle$ & existing link $\langle ij\rangle$ is deleted from $\mathfrak{g}$. \\ \hline
$\mathcal{G}(N)$ & the set of all networks on $N$ agents.\\ \hline
$\mathfrak{g}(\kappa_{i})$ & {a component of network $\mathfrak{g}$, where $\kappa_{i}$ is the set of agents in that component.} \\ \hline
{$\mathfrak{g^c}$} & {complement of network $\mathfrak{g}$.} \\ \hline
\end{tabular}
\label{tab:notation-summary}
\end{table}

Given a social storage network $\mathfrak{g} =(\mathbf{A},\mathbf{L})$, the link $\langle ij\rangle \in \mathbf{L}$ represents the fact that agents $i$ and $j$ are neighbors of each other, and are involved in a data backup partnership. This partnership indicates that both the agents commit to share their storage resources with each other so that they can backup their data on each other's shared storage space. Storage resource sharing and data backup activity are bidirectional and occur with the mutual consent of $i$ and $j$. This implies, the link $\langle ij\rangle$ and the link $\langle ji\rangle$ are identical. We also refer to $i$ and $j$ as backup partners. 
The set of agents with whom agent $i$ has links is represented by $\eta_{i}(\mathfrak{g})$. In other words, $\eta_{i}(\mathfrak{g})$ is the neighborhood of agent $i$. We also use $\eta_{i}(\mathfrak{g})$ to represent the neighborhood size of $i$, which will be clear from the context. 

At any given point in time, each agent plays a dual role: that of a data owner who wants to back up its data, and that of storage provider who provides storage space for each of its backup partners. Pairs of agents may add a new link (or continue to maintain the existing link) or delete the existing link (or continue to remain without a link). 
In the context of social storage, mutual consent is necessary for adding as well as for deleting links. 
That is, \textit{an agent does not add a new link without the consent of the agent with whom it wants to add the new link and does not delete an existing link without the consent of the agent from whom it wants to delete the existing link}. 

The structure of the network, $\mathfrak{g}$, is determined by actions of the agents. Firstly, the network is updated when two agents $i$ and $j$ add a new link $\langle ij\rangle$, and we denote this by $\mathfrak{g}+\langle ij\rangle$. Secondly, the network is updated when a pair agents $i$ and $j$ delete an existing link $\langle ij\rangle$, and we denote this by $\mathfrak{g}-\langle ij\rangle$. As agents themselves decide with whom they want to perform backup partnerships and with whom they do not, this is a process of endogenous network formation (or partner selection). 
{In this paper, we do not explicitly consider trust between pairs of agents. We assume that {\it all agents} trust each other, and thus, anyone can form links with anyone.}

\sloppy
\par A social storage network may be connected or may consist of two or more connected components. We say a network $\mathfrak{g}$ is connected if there exists a path between every pair of {agents} $i$ and $j \in \mathbf{A}$, or else the network $\mathfrak{g}$ is disconnected. A disconnected network $\mathfrak{g}$ can be partitioned into disjoint sub-networks $\mathfrak{g}(\kappa_1), \mathfrak{g}(\kappa_2),...,\mathfrak{g}(\kappa_n)$, where $\kappa_1 \cup \kappa_2 \cup \ldots \cup \kappa_{n} = \mathbf{A}, \kappa_{r} \cap \kappa_{s} = \phi$ for all $r, s \in \{1, 2, \ldots, n\}, r \neq s$, such that any pair of {agents} $i$ and $j$ are connected if and only if $i$ and $j$ are elements of the same set $\kappa_{r}$. Such sub-networks are called as \textit{components} of the network $\mathfrak{g}$. 

A \textit{complete} network is one where every agent is connected to every other agent. A \textit{null (or empty)} network is one where there are no links --- that is, no agent is connected to any agent. A component which is complete is called a \textit{clique}. 

An \textit{$r$-regular} network is one where each agent has exactly $r$ neighbors. An \textit{$N$ agent star} network consists of a single universal agent and $N-1$ pendant agents. A universal agent is one who is adjacent to other $N-1$ pendant agents. A pendant agent is one who is adjacent to only the universal agent. A star component is a component which is a star (sub-)network. 

The \textit{complement} of network $\mathfrak{g}$, denoted by $\mathfrak{g^c}$, is a network on the same set of agents such that $\langle ij \rangle \in \mathfrak{g^c}$ if and only if $\langle ij \rangle \not \in \mathfrak{g}$.

Table \ref{tab:notation-summary} summarizes all notations used in this paper.


\subsection{Utility of an Agent in a Social Storage Network}\label{sec:frameworks}
Data stored on local hard disk is in danger of getting lost or damaged due to local disk failure. Hence, to keep data safe, each data owner wants to backup its data. 
Social storage systems use two types of techniques to backup data. The first is erasure coding, and the second is replication\footnote{\cite{erasure-vs-replication} perform quantitative comparisons between these two techniques.} \citep{f2frecoverability}. Erasure coding is the data redundancy technique in which a data object is divided into $x$ blocks and recoded into $y$ blocks ($y>x$). Then the main data block can be recovered from any subset of $y$. Replication is the data redundancy technique in which an agent maintains a single data copy on each partner's storage device. In this paper, we consider the replication technique.
As hard disks are prone to failure, there is a chance that a data owner's backup partner's hard disk also fails. It is likely that each backup partner's hard disk fails, so each data owner's interest lies in recovering at least one copy of its data so that the value of the data is intact. It is not hard to observe that each agent's chance of data recovery,  given a particular disk failure rate, depends on its neighborhood size. The more the number of neighbors, the higher the chance of data recovery. 
\par In the absence of costs to add and maintain links, the aim of each agent in a social storage network is to maximize {the} chance of data recovery, given that the local copy of data has been damaged or lost. However, every agent incurs a cost for each of its links. Keeping this in mind, we define the utility of each agent in the network under two frameworks. The utility of agent $i$ in the network $\mathfrak{g}$ is represented by a function  $u_{i}:\mathcal{G}(\mathbf{A})\rightarrow \mathbb{R}$, where $\mathcal{G}$ is the set of all networks, ($\mathfrak{g}$ is an element of $\mathcal{G}$). The profile of utility functions $(u_1,. . ., u_n)$ is a vector of utilities for all agents. We first define the parameters required to define the utility function. $\lambda \in (0,1)$ is the average disk failure rate in the network. That is, at any point in time, the probability of failure of agent $i$'s disk is $\lambda$. For data owner (agent) $i$, the value of the local data that is to be backed up, is $\beta_{i}$. Each agent incurs a cost $c$ to maintain a link. That is, the total cost of adding and/ or maintaining a link is $2c$, and we assume that the agents connected by the link equally share this cost. This cost can be interpreted as the cost required for infrastructure, bandwidth, time, etc. There is no {\em{additional}} cost to add a new link. Each agent $i$ also has allocated budget $b_{i}$ for maintaining its links. Further, each agent $i$ has a certain amount of local data $d_{i}$ that the agent wants to store on  storage devices of backup partners. Also, each agent $i$ has a certain amount of storage space $s_{i}$ available for sharing with other agents in the network. Using these parameters, we now define the utility of an agent in the following two frameworks.

\subsubsection{Multi-Objective Framework (MO-Framework)}\label{subsub:framework1}
In the first framework, there are two objective functions that each agent $i$ tries to optimise. Firstly, each agent $i$ wants to minimise the total cost associated with maintaining the links, i.e., $c\eta_i(\mathfrak{g})$. Secondly, each agent wants to maximise the expected value of backup data. Since the disk failure rate is $\lambda$, and $i$ has $\eta_i(\mathfrak{g})$ neighbors, the expected value of $i$'s backup data is $\beta_{i}(1-\lambda^{\eta_i(\mathfrak{g})})$. 
Note that, as each agent is interested in ``how many links to maintain'', we look at the expected value of an agent's backup data {\em{given}} that the local copy of the agent's data has been damaged or lost. 
For each agent $i$, these two objective functions can be written as a single objective function as follows: 
\begin{equation}\label{eq:objectivefunction}
[\alpha( \beta_{i}(1-\lambda^{\eta_{i}(\mathfrak{g})}))]-[(1-\alpha) (c\eta_i(\mathfrak{g}))], \quad \text{where} \quad \alpha\in (0,1).
\end{equation}
For elegance of results on stability, we let $\alpha=1/2$. We drop the factor of $1/2$ from (\ref{eq:objectivefunction}), for all $i \in \mathbf{A}$, and just consider the following utility function $u_{i}(\mathfrak{g})$, for all $i \in \mathbf{A}$, for the given network $\mathfrak{g}$:
\begin{equation}\label{eq:utility-mo-frame}
u_{i}(\mathfrak{g})=\beta_{i}(1-\lambda^{\eta_{i}(\mathfrak{g})})-c \eta_{i}(\mathfrak{g}).
\end{equation}

\textcolor{mycolor}{As evident above, this is no longer a MO-optimization problem. We have done this conversion because (a) this is one of the easiest way to solve a MO-problem, and (b) our focus is on the network formation game, stability, efficiency, and contentedness of the network. Solving the MO-optimization problem without this conversion is part of future work, and we discuss that in Section 7. We also still call this a MO-framework a nomenclature (to differentiate with Single Objective (SO)-framework discussed below).}  

Each agent $i$ wants to maximise $u_i(\mathfrak{g})$ over all possible values of $\eta_i(\mathfrak{g})$. The social optimisation problem can be formulated as
\begin{center}
${max\ }\sum\limits_{i \in \mathbf{A}}(u_{i}(\mathfrak{g}))$\\ \quad \\
\textit{{such that}} \\ \quad \\
$\eta_{i}(\mathfrak{g})=\sum\limits_{i,j\in \mathfrak{g}}^{}a_{ij}$ and \\

$ s_{i}\geq \sum\limits_{j\in \eta_{i}(\mathfrak{g})}^{}d_{j}a_{ij},$\\ \quad \\ 
where,
\[ a_{ij} =
  \begin{cases}
    1       & \quad \text{if } i \text{  and } j \text{  have a backup agreement},\\
    0  & \quad \text{otherwise}.\\
  \end{cases}
\]
\end{center}

\subsubsection{Single Objective Framework (SO-Framework)}\label{subsub:framework2}
In this framework, each agent $i$ has only one objective (as compared to two in the previous framework). 
Each agent tries to maximise the expected value of backup data. 
The cost, $c\eta_i(\mathfrak{g})$, incurred by agent $i$ to maintain links (which was 
the second objective function in the MO-Framework), appears in constraints here. This is because, in the SO-Framework, every agent $i$ has 
an allocated budget, $b_i$, towards backup agreements. 

For the given network  $\mathfrak{g}$, utility $u_{i}(\mathfrak{g})$ of agent $i$ is
 
$$u_{i}(\mathfrak{g})=\beta_{i}(1-\lambda^{\eta_{i}(\mathfrak{g})}).$$

Each agent $i$ wants to maximise $u_i(\mathfrak{g})$ over all possible values of $\eta_i(\mathfrak{g})$. The social optimisation problem can be formulated as
\begin{center}
${max\ }\sum\limits_{i \in \mathbf{A}}(u_{i}(\mathfrak{g}))$\\ \quad \\
\textit{{such that}}\\ \quad \\
$\eta_{i}(\mathfrak{g})=\sum\limits_{i,j\in \mathfrak{g}}^{}a_{ij},$ \\ 
$s_{i}\geq \sum\limits_{j\in \eta_{i}(\mathfrak{g})}^{}d_{j}a_{ij}, and$  \\ 
\text{}\\ \text{}\\ 
$b_{i}\geq c\eta_{i}(\mathfrak{g}),$\\ \quad \\
where,
\[ a_{ij} =
  \begin{cases}
    1       & \quad \text{if } i \text{  and } j \text{  have a backup agreement},\\
    0  & \quad \text{otherwise}.\\
  \end{cases}
\]
\end{center}

\begin{remark}
{The utility function in the SO-Framework may be reduced to the Constant Absolute Risk Aversion (CARA \cite{Pratt-CARA})}\footnote{{We refer  the readers to a survey by \cite{Meyer-CARA-2010} on functional forms for the utility functions of agents, based on their risk taking abilities.}} function. In the context of social storage, agents are {\em{risk averse}} as they do not want to ``risk'' losing their data, which is what the above utility function captures. This function may also be viewed as the Cumulative Distribution Function of an Exponential distribution, given that the disk failure rate is Poisson.
\end{remark}

\begin{remark}
We explicitly write the formulation of the social optimization problems in the two different frameworks, as above, primarily to highlight that the cost is moved from the utility function in the MO-framework to budget constraints in the SO-framework. Our goal is not to solve these problems but rather analyze the corresponding network properties, for example, the efficiency of the resulting networks.   
\end{remark}

\subsection{Bilateral Stability, Efficiency and Contentedness} \label{sec:background} 
There is a need for a solution concept which is suitable for characterizing the storage network formation game. A strategic network formation game (NFG) is described as below. NFG consists of a set of agents $\mathbf{A}=\{1, 2,...,N\}$ who represent nodes in the network $\mathfrak{g}$ --- if $i$ is an agent, we use $i \in \mathbf{A}$ and $i \in \mathfrak{g}$ synonymously. In this setting, pairs of agents may form new links thereby increasing their expected value of backup data, by incurring higher costs to maintain links. Pairs of agents may also delete existing links, thereby reducing the costs incurred, but reducing the probability of retrieving the data too. 
The shape of the network is not only defined by each agent's cost and benefit trade off, but also by limitation of resources available with the agents. 

\par {\emph{Pairwise stability}} introduced by \cite{jackson} (see Definition \ref{def:jackson:stability}) is an appropriate solution concept  when agents require mutual consent while adding a link, but any agent can delete any of its existing links without consent. 

\begin{definition}\label{def:jackson:stability}\citep{jackson}
	\textit{A network $\mathfrak{g}$ is pairwise stable if and only if}
\begin{enumerate}
	\item for all $\langle ij\rangle \in \mathfrak{g}$, $u_{i}(\mathfrak{g})\geq u_{i}(\mathfrak{g}-\langle ij\rangle)$ and  $u_{j}(\mathfrak{g})\geq u_{j}(\mathfrak{g}-\langle ij\rangle)$, and
	\item for all $\langle ij\rangle \not \in \mathfrak{g}$, if $u_{i}(\mathfrak{g}+\langle ij\rangle) > u_{i}(\mathfrak{g})$, then  $u_{j}(\mathfrak{g}+\langle ij\rangle)< u_{j}(\mathfrak{g})$.
\end{enumerate}
\end{definition}

\par We modify the pairwise stability concept introduced by \cite{jackson} so as to ensure that deletion of links also happens with mutual consent. We call this modified pairwise stability as {\emph{bilateral stability}}. 

{\emph{Bilateral equilibrium}} \citep{Goyal-bilateral-deviation-important-paper} is another refinement of {\emph{pairwise stability}} \citep{jackson}. 
\cite{Goyal-bilateral-deviation-important-paper} define strategies of agents as sets of links they would want to add, and define {\emph{bilateral equilibrium}} as a strategy profile that is a Nash equilibrium (that is, no agent benefits by unilaterally deviating) and is pairwise stable (where both addition and deletion require mutual consent). 
The set of all bilaterally stable strategies (see Definition \ref{def:stability}) is a superset of the set of all bilateral equilibrium strategies \citep{Goyal-bilateral-deviation-important-paper}, as discussed earlier.  

The modified definition of pairwise stability we use for social storage is given below.
\begin{definition}\label{def:stability}
\textit{A social storage network $\mathfrak{g}$ is bilaterally stable if and only if}
\begin{enumerate}
\item for all $\langle ij\rangle \in \mathfrak{g}$, if $u_{i}(\mathfrak{g}-\langle ij\rangle)>u_{i}(\mathfrak{g})$, then $u_{j}(\mathfrak{g}-\langle ij\rangle)< u_{j}(\mathfrak{g})$, and
\item for all $\langle ij\rangle \not \in \mathfrak{g}$, if $u_{i}(\mathfrak{g}+\langle ij\rangle) > u_{i}(\mathfrak{g})$, then $u_{j}(\mathfrak{g}+\langle ij\rangle) < u_{j}(\mathfrak{g})$.
\end{enumerate}
\end{definition}

Definition \ref{def:stability} is a network stability concept, whose first part states that no pair of agents with a link between them, wants to delete the link, and the second part states  that no pair of agents has an incentive to add a new link. Note that neither link formation (addition) nor link deletion can happen without mutual consent. Our further discussions about social storage stability stands on Definition \ref{def:stability}.

\begin{remark}
\textcolor{mycolor}{Definition \ref{def:stability} can be rewritten using only conditions on the addition of links by rewriting the first condition (that is, the deletion condition) as a condition for addition in $\mathfrak{g^c}$. Similarly, we can also rewrite Definition \ref{def:stability} using only deletion conditions.}
\end{remark}



Now, we generalize Definition \ref{def:stability} so that it is suitable as a solution concept for the two frameworks discussed in the previous section.

For this, we first define {\em{remaining storage}} available with agent $i$ in a network $\mathfrak{g}$ as 

\begin{equation}\label{eq:remaining-storage}
RS_{i}= s_{i} - \sum\limits_{j\in \eta_{i}(\mathfrak{g})}^{}d_{j}a_{ij},
\end{equation}
\noindent and {\em{remaining budget}} of agent $i$ in $\mathfrak{g}$ as 

\begin{equation}\label{eq:remaining-budget}
RB_{i}=b_{i} - \sum\limits_{j\in \eta_{i}(\mathfrak{g})}^{}c a_{ij},
\end{equation}

\noindent where 

\[ a_{ij} =
  \begin{cases}
    1       & \quad \text{if } i \text{  and } j \text{  have a backup agreement},\\
    0  & \quad \text{otherwise}.\\
  \end{cases}
\]

For the MO-Framework, where we have storage constraints, the following modification of Definition \ref{def:stability} is appropriate. 

\begin{definition}\label{def:stability-storage-constraints}
\textit{A social storage network $\mathfrak{g}$ with storage constraints is bilaterally stable if and only if}
\begin{enumerate}
\item for all $\langle ij\rangle \in \mathfrak{g}$, if $u_{i}(\mathfrak{g}-\langle ij\rangle)>u_{i}(\mathfrak{g})$, then $u_{j}(\mathfrak{g}-\langle ij\rangle)< u_{j}(\mathfrak{g})$, and 
\item for all $\langle ij\rangle \not \in \mathfrak{g}$, if $[u_{i}(\mathfrak{g}+\langle ij\rangle)  > u_{i}(\mathfrak{g})$ and $RS_{j}\geq d_{i}]$, then

 \quad \quad \quad \quad \quad $[u_{j}(\mathfrak{g}+\langle ij\rangle) < u_{j}(\mathfrak{g})$  or $RS_{i}<d_{j}]$. 
\end{enumerate}
\end{definition}

In the above definition, there is no change in the link deletion condition of Definition \ref{def:stability}. However, while adding a link, an agent has to ensure that the other agent has sufficient storage to store its data (besides ensuring increase in its utility). We assume that the agents are rational and self-centered (and hence, it is up to agent $j$ to check whether agent $i$ has sufficient storage for agent $j$'s data or not). 

Next, we adapt Definition \ref{def:stability} for the SO-Framework, where we have storage and budget constraints. 

\begin{definition}\label{def:stability-storage-budget-constraints}
\textit{A social storage network $\mathfrak{g}$ with storage and budget constraints is bilaterally stable if and only if}
\begin{enumerate}
\item for all $\langle ij\rangle \in \mathfrak{g}$, if $u_{i}(\mathfrak{g}-\langle ij\rangle)>u_{i}(\mathfrak{g})$, then $u_{j}(\mathfrak{g}-\langle ij\rangle)< u_{j}(\mathfrak{g})$, and 
\item for all $\langle ij\rangle \not \in \mathfrak{g}$, if $[u_{i}(\mathfrak{g}+\langle ij\rangle)  > u_{i}(\mathfrak{g})$ and $RS_{j}\geq d_{i}$ and $RB_{i}\geq c)]$, then

 \quad \quad \quad \quad  \quad $[u_{j}(\mathfrak{g}+\langle ij\rangle) < u_{j}(\mathfrak{g})$  or $RS_{i}<d_{j}$ or $RB_{j}<c]$. 
\end{enumerate}
\end{definition}

As in the case of MO-Framework, there is no change in the link deletion condition of Definition \ref{def:stability}. However, while adding a link, an agent has to ensure that the other agent has sufficient storage to store its data and agent itself has sufficient budget to form the link (besides ensuring increase in its utility). This is, again, based on the assumption that the agents are rational and self-centered. 

We, now, define efficient and contented social storage networks, with as well as without constraints. Efficient social storage networks are social storage networks where as many agents as possible achieve maximum utility, whereas contented social storage networks are those where all agents achieve maximum utility.
\begin{definition}\label{def:welfare}
\textit{A social storage network $\mathfrak{g}$ is efficient with respect to utility profile $(u_1,..., u_N)$ if $\sum\limits_{i}u_i(\mathfrak{g}) \geq \sum\limits_{i}u_i(\mathfrak{g}')$, for all $\mathfrak{g}' \in \mathcal{G}(N)$}.
\end{definition}
\begin{definition}\label{def:welfare-storage-constraint}
\textcolor{mycolor}{\textit{A social storage network $\mathfrak{g}$ with storage constraints is efficient with respect to utility profile $(u_1,..., u_N)$ if $\sum\limits_{i}u_i(\mathfrak{g}) \geq \sum\limits_{i}u_i(\mathfrak{g}')$, for all $\mathfrak{g}' \in \mathcal{G}(N)$ where $RS_{i} \geq 0$ for all $i \in \mathfrak{g}'$ }.}
\end{definition}
\begin{definition}\label{def:welfare-storage-budget-constraints}
\textcolor{mycolor}{\textit{A social storage network $\mathfrak{g}$ with storage and budget constraints is efficient with respect to utility profile $(u_1,..., u_N)$ if $\sum\limits_{i}u_i(\mathfrak{g}) \geq \sum\limits_{i}u_i(\mathfrak{g}')$, for all $\mathfrak{g}' \in \mathcal{G}(N)$ where $RS_{i} \geq 0$ and $RB_{i} \geq 0$, for all $i \in \mathfrak{g}'$ }.}
\end{definition}

\begin{definition}\label{def:welfare-contented}
A social storage network $\mathfrak{g}$ is contented with respect to utility profile $(u_1,..., u_N)$ if, for each $i \in A$, $u_i = \underset{\eta_{i}(\mathfrak{g})} {max\ }(\beta_{i}(1-\lambda^{\eta_{i}(\mathfrak{g})})-c \eta_{i}(\mathfrak{g}))$, under the MO-Framework, 
and $u_i = \underset{\eta_{i}(\mathfrak{g})} {max\ }(\beta_{i}(1-\lambda^{\eta_{i}(\mathfrak{g})}))$, under the SO-Framework. 
\end{definition}

\textcolor{mycolor}{
\begin{remark}
If maximum possible utility is not achievable by a one or more agents because of storage or budget constraints, then those agents are not contented, and hence, the social storage network is not contented. Therefore, we do not define contentedness with constraints.
\end{remark}
}

\section{Stable Network Characterization and Stability Point}\label{sec:stability-point}
In this section, we study the stability aspects of social storage networks considering the utilities of agents and the solutions concept as defined in Sections \ref{sec:frameworks} and \ref{sec:background}, respectively.

Free riding is a situation where an agent offers less storage space, but consumes more. To deal with free riding, many backup systems have used the concept of symmetric resource sharing (or equal resource trading). Internet Cooperative Backup System \citep{cbssystem}, PeerStore \citep{peerstore-system}, Pastiche \citep{pastiche}, are a few examples of P2P backup systems, which use symmetric resource trading to mitigate free riding.

We term a social storage system with symmetric resource sharing as a symmetric social storage system. We consider symmetry in the agents'  value of their respective data, storage space available, amount of data to be shared, and budget in two different scenarios. These scenarios are discussed next.

\begin{definition}\label{def:symmetric-sv-network}
A symmetric value network (SVN) $\mathfrak{g}$ is a social storage network where the benefit (value) associated with backed-up data is the same for all agents in the network, i.e., $\beta_{i} = \beta_{j}$  (say  $\beta$), for all $i, j \in \mathbf{A}$, and hence, utility of each agent $i$ in the network is
\begin{center}
\begin{equation}\label{eq:utility}
\begin{split}
u_{i}(\mathfrak{g})=\beta(1-\lambda^{\eta_{i}(\mathfrak{g})})-c \eta_{i}(\mathfrak{g}) \text{\ for MO-Framework \ref{subsub:framework1} and},  \\
u_{i}(\mathfrak{g})=\beta(1-\lambda^{\eta_{i}(\mathfrak{g})}) \text{ for SO-Framework \ref{subsub:framework2}},\\
\end{split}
\end{equation}
where  $\beta,\lambda, c \in (0,1)$. 
\end{center}
\qed
\end{definition}


\begin{definition}\label{def:symmetric-sb-network}
A symmetric resource network (SRN) $\mathfrak{g}$ is a social storage network where all agents in $\mathfrak{g}$ have an equal amount of (limited) storage space available to them, an equal amount of data that they want to backup, and have the same budget. That is, for all $i, j\in \mathfrak{g}$, $s_{i} = s_{j}$ (say $s$),  $d_{i} = d_{j}$ (say $d$), and $b_{i} = b_{j}$ (say $b$).
\end{definition}

\begin{remark}
\textcolor{mycolor}{From this symmetric setup, we can move to real life scenarios in many ways. We can have different value of cost and benefit for different agents. Another way to include heterogeneity in this model is by using the concept of Social Range Matrix \citep{stefan-social-range-matrix}, which we have done recently \citep{jain-social-storage}. Here, each agent is concerned about its perceived utility, which is a linear combination of its utility as well as others utilities (depending upon whether the pair are friends, enemies or do not care about each other).}
\end{remark}

Now, we work with SVN under the MO-Framework, where each agent in the given network $\mathfrak{g}$ has as much storage as is required for all other agents in $\mathfrak{g}$. That is, 
\begin{equation}\label{eq:SVN-enough-storage}
s_{i} \geq \sum\limits_{\substack{j\in \mathfrak{g}, \\ j\not=i}}^{} d_{j}, \quad \mbox{for all\ \ } i \in \mathfrak{g}.
\end{equation}
Note that $s_{i}$ may be different from some other $s_{j}$. For convenience, we shall call such a network as SVN with sufficient storage. The reason we do this is that it leads to the results of the realistic scenario, that is, SV-SRN under the MO-Framework. 

\begin{remark}
An SV-SRN, $\mathfrak{g}$ under the MO-Framework is a social storage network where the utility of each agent $i\in \mathfrak{g}$ is $u_{i}(\mathfrak{g})=\beta_{i}(1-\lambda^{\eta_{i}(\mathfrak{g})})-c \eta_{i}(\mathfrak{g})$, and for all agents $i, j \in \mathfrak{g}$, $\beta_i = \beta_j, s_i = s_j$, and $d_i = d_j$.
\end{remark}

Next, we work with SVN under the SO-Framework where each agent in the given network $\mathfrak{g}$ has as much storage as is required for all other agents in $\mathfrak{g}$, and 
each agent in the given network $\mathfrak{g}$ has as much budget as is required to maintain backup-partnerships with every other agent in $\mathfrak{g}$. 
That is, 
\begin{equation}\label{eq:SRN-enough-budget}
s_{i}\geq \sum\limits_{j\in A, \\ j\not=i}^{}d_{j}, \quad \mbox{for all\ \ } i \in \mathfrak{g}, \quad \mbox{and},  \quad  b_{i}\geq c(N-1),  \quad \mbox{for all\ \ } i \in \mathfrak{g} \quad \text{where } N = |\mathbf{A}|.  
\end{equation}

As in the SO-Framework, this leads to the scenario of SV-SRN under the SO-Framework. However, for SO-Framework, we present the results for SRN directly rather than SV-SRN. This is because SV-SRN is a subset of SRN and so, what holds for SRN does for SV-SRN as well.

\begin{remark}
An SV-SRN, $\mathfrak{g}$ under the SO-Framework is a social storage network where the utility of each agent $i\in \mathfrak{g}$ is $u_{i}(\mathfrak{g})=\beta_{i}(1-\lambda^{\eta_{i}(\mathfrak{g})})$, and for all agents $i, j \in \mathfrak{g}$, $\beta_i=\beta_j, s_i=s_j, d_i = d_j$, and $b_i=b_j$.
\end{remark}

For ease of exposition, from now onwards, whenever we discuss SVN networks, we will always assume sufficiency of every resource --- that is, sufficient storage under MO-Framework, and sufficient storage and budget under SO-Framework. Whenever we discuss SRN or SV-SRN networks, we will not make these assumptions of sufficiency. These are summarised in Table  \ref{table:Different-Networks}.

\begin{table}[h!]
\centering
\caption{Summary of Network Study under Different Frameworks with/ without Sufficient Resources.}
\begin{tabular}{lll}
\hline
\textbf{Network Type} & \textbf{Framework} & \textbf{Resource Availability} \\ 
\hline
&&\\
SVN & MO-Framework & Sufficient Storage.\\ 
&&\\
SV-SRN & MO-Framework & Limited Storage and Limited Budget.\\
&&\\
SVN & SO-Framework & Sufficient Storage and Sufficient Budget.\\ 
&&\\
SRN & SO-Framework & Limited Storage and Limited Budget.\\
&&\\
\hline
\end{tabular}
\label{table:Different-Networks}
\end{table}



In the following subsections, we characterize bilaterally stable symmetric social storage networks, by first deriving the deviation conditions --- conditions for an agent to have an incentive to add or delete a link, given the network parameters (that is, disk failure rate $\lambda$, value of backup data $\beta$, and the cost of maintaining a link $c$). This also gives us necessary and sufficient conditions for bilateral stability, in terms of the network parameters ($\lambda$, $\beta$, and $c$). Further, this  makes it easier to visualize a bilaterally stable network, and we use these conditions to derive the ideal neighborhood size for having a bilaterally stable network. We term this ideal neighborhood size as the {\it stability point} (see Definition \ref{def:stability-point}).

\begin{definition}\label{def:stability-point}
Given a network $\mathfrak{g}$, we define the \textit{stability point} $\hat{\eta}$ of $\mathfrak{g}$ as the neighborhood size (degree) such that no agent in $\mathfrak{g}$ has any incentive to increase its neighborhood size to more than $\hat{\eta}$ and to decrease it to less than $\hat{\eta}$.
\end{definition}

{We, now, characterize SVN and SV-SRN under the MO-Framework, and SVN and SRN under the SO-Framework. Further, we prove uniqueness of the {\emph{stability point}} of these networks and also show that the stability point is independent of the number of agents for all cases under the MO-Framework and for all cases but one trivial case under the SO-Framework, the trivial case being SVN with sufficient storage and sufficient budget where it is easy to see that the complete network is the only stable network.}

\subsection{Characterization Under the MO-Framework}
In this subsection, we characterize bilaterally stable SVN and SV-SRN under the MO-Framework. We, first, derive conditions under which an agent has an incentive to add a new link or delete an existing link. Then, we derive necessary and sufficient conditions for bilateral stability of SVN and SV-SRN under the MO-Framework, and prove that the stability point of these networks is unique and independent of the number of agents.
\begin{lemma}\label{lemma:add}
\textcolor{mycolor}{In an SVN $\mathfrak{g}$, under the MO-Framework, for any agent $i\in \mathfrak{g}$, forming a partnership with another agent $j\in \mathfrak{g}$ is beneficial if and only if $c < \beta \lambda^{\eta_{i}(\mathfrak{g})}[1-\lambda]$.}
\end{lemma}
\begin{proof}
As $\mathfrak{g}$ is an SVN under the MO-Framework, 

$u_{i}(\mathfrak{g})=\beta(1-\lambda^{\eta_{i}(\mathfrak{g})})-c\eta_{i}(\mathfrak{g})$, for all $i \in \mathfrak{g}$.\\

For $i,j \in \mathfrak{g}$, if $\langle ij \rangle \not \in \mathfrak{g}$, then

 $u_{i}(\mathfrak{g}+\langle ij\rangle)=[\beta(1-\lambda^{\eta_{i}(\mathfrak{g})+1})]-[c (\eta_{i}(\mathfrak{g})+1)]$.\\

Adding a new link or backup partner is beneficial for $i$ if and only if 

$u_{i}(\mathfrak{g}+\langle ij \rangle)> u_{i}(\mathfrak{g})$, if and only if

$[\beta(1-\lambda^{\eta_{i}(\mathfrak{g})+1})-c(\eta_{i}(\mathfrak{g})+1)] > [\beta(1-\lambda^{\eta_{i}(\mathfrak{g})})-c(\eta_{i}(\mathfrak{g}))]$, if and only if

$c < \beta[\lambda^{\eta_{i}(\mathfrak{g})}-\lambda^{\eta_{i}(\mathfrak{g})+1}]$.	\qed
\end{proof}

\begin{remark}\label{remark:add}

\textcolor{mycolor}{The term on the left-hand side of the inequality in Lemma \ref{lemma:add} is the cost $c$ that agent $i$ incurs in order to add a new neighbour $j$. The term on the right-hand side is the expected benefit that agent $i$ receives by forming a new link with neighbour $j$.}
\end{remark}
\begin{lemma}\label{lemma:delete}
\textcolor{mycolor}{In an SVN $\mathfrak{g}$, under the MO-Framework, for any agent $i\in \mathfrak{g}$, breaking an existing partnership with another agent $j\in \mathfrak{g}$ is beneficial if and only if 
\\ $c > \beta \lambda^{\eta_{i}(\mathfrak{g})-1}[1-\lambda]$.}
\end{lemma}
\begin{proof}
As $\mathfrak{g}$ is an SVN, under the MO-Framework, $u_{i}(\mathfrak{g})=\beta(1-\lambda^{\eta_{i}(\mathfrak{g})})-c\eta_{i}(\mathfrak{g})$, for all $i \in \mathfrak{g}$.\\

If $\langle ij \rangle \in \mathfrak{g}$, then $u_{i}(\mathfrak{g}-\langle ij\rangle)=[\beta(1-\lambda^{\eta_{i}(\mathfrak{g})-1})]-[c (\eta_{i}(\mathfrak{g})-1)]$.\\

Deleting an existing link is beneficial for any agent $i$ if and only if 

$u_{i}(\mathfrak{g}-\langle ij \rangle)> u_{i}(\mathfrak{g}))$, if and only if 

$[\beta(1-\lambda^{\eta_{i}(\mathfrak{g})-1})-c(\eta_{i}(\mathfrak{g})-1)] > [\beta(1-\lambda^{\eta_{i}(\mathfrak{g})})-c(\eta_{i}(\mathfrak{g}))]$, if and only if

$c > \beta[\lambda^{\eta_{i}(\mathfrak{g})-1}-\lambda^{\eta_{i}(\mathfrak{g})}]$.							\qed
\end{proof}

\begin{remark}
\textcolor{mycolor}{We interpret Lemma \ref{lemma:delete} in lines similar to Remark \ref{remark:add}.}
\end{remark}

\begin{theorem}\label{theorem:stability}
An SVN $\mathfrak{g}$, under the MO-Framework, is bilaterally stable if and only if
\begin{enumerate}
\item \textcolor{mycolor}{for all $\langle ij\rangle \in \mathfrak{g}$, if $\beta \lambda^{\eta_{i}(\mathfrak{g})-1}[1-\lambda]< c$,  then $\beta \lambda^{\eta_{j}(\mathfrak{g})-1}[1-\lambda]>c$, and
\item for all $\langle ij\rangle \not \in \mathfrak{g}$, if $\beta \lambda^{\eta_{i}(\mathfrak{g})}[1-\lambda]>c$, then $\beta\lambda^{\eta_{j}(\mathfrak{g})}[1-\lambda]<c$.}
\end{enumerate}
\end{theorem}
\begin{proof}
Follows from Lemma \ref{lemma:add}, Lemma \ref{lemma:delete} and Definition \ref{def:stability} of bilateral stability.	\qed
\end{proof}

We state and prove the following for SV-SRN, under the MO-Framework.

\begin{lemma}\label{lemma:adding-deleting-link-beneficial-in-SV-SRN-MO-Framework}
Let $\mathfrak{g}$ be an SV-SRN, under the MO-Framework. For any agent $i\in \mathfrak{g}$, adding a new partnership with agent $j\in \mathfrak{g}$ is beneficial if and only if \\
\textcolor{mycolor}{$(c < \beta \lambda^{\eta_{i}(\mathfrak{g})}[1-\lambda] \quad and \quad s-d\eta_{j}(\mathfrak{g})\geq d)$, \\
and breaking an existing partnership with agent $j\in \mathfrak{g}$ is beneficial if and only if \\
 $c > \beta \lambda^{\eta_{i}(\mathfrak{g})-1}[1-\lambda]$.}
\end{lemma}
\begin{proof}
If $\langle ij\rangle \not \in \mathfrak{g}$, agent $i$ has an incentive to add a link with agent $j$, if and only if \\

\noindent $c < \beta[\lambda^{\eta_{i}(\mathfrak{g})}-\lambda^{\eta_{i}(\mathfrak{g})+1}]$ (from Lemma \ref{lemma:add}), where $\eta_i(\mathfrak{g}) = $ neighborhood size of $i$, 

\quad \quad \quad \quad \quad \quad \quad \quad and the amount of storage available with agent $j \geq$ agent $i$'s 

\quad \quad \quad \quad \quad \quad \quad \quad data size, if and only if\\

\noindent $ c < \beta[\lambda^{\eta_{i}(\mathfrak{g})}-\lambda^{\eta_{i}(\mathfrak{g})+1}]$ and $s_j-\sum\limits_{k \in \eta_j(\mathfrak{g})}{d_k} \geq d_i$, 

\quad \quad \quad \quad \quad \quad \quad \quad where $\eta_j(\mathfrak{g})$ is the set of neighbors of $j$, if and only if \\

\noindent $ c < \beta[\lambda^{\eta_{i}(\mathfrak{g})}-\lambda^{\eta_{i}(\mathfrak{g})+1}]$ and $s-d\eta_{j}(\mathfrak{g})\geq d$, (as $s_k = s_l, d_k = d_l$, for all $k, l \in \mathfrak{g}$), 

\quad \quad \quad \quad \quad \quad \quad \quad where $\eta_j(\mathfrak{g})$ is the neighborhood size of $j$. \\

To delete an existing link, agent $i$ only looks at the cost for maintain the link, and hence, from Lemma \ref{lemma:delete}, agent $i$ has an incentive to delete a link if and only if \\
 $c > \beta[\lambda^{\eta_{i}(\mathfrak{g})-1}-\lambda^{\eta_{i}(\mathfrak{g})}]$.	\qed
\end{proof}

\begin{theorem}\label{theorem:stability-SV-SRN-MO-Framework}
An SV-SRN $\mathfrak{g}$, under the MO-Framework, is bilaterally stable if and only if
\end{theorem}
\begin{enumerate}
\item for all $\langle ij\rangle \in \mathfrak{g}$, $\beta \lambda^{\eta_{i}(\mathfrak{g})-1} [1-\lambda]< c \Rightarrow \beta \lambda^{\eta_{j}(\mathfrak{g})-1}[1-\lambda]>c$, and
\item for all $\langle ij\rangle \not \in \mathfrak{g}$, $\beta \lambda^{\eta_{i}(\mathfrak{g})}[1-\lambda] >c$ and $s-d\eta_{j}(\mathfrak{g})\geq d \Rightarrow$ \\
\quad $\beta \lambda^{\eta_{j}(\mathfrak{g})}[1-\lambda]<c$ or $s-d\eta_{i}(\mathfrak{g})< d$.
\end{enumerate}
\begin{proof}
Follows from Lemma \ref{lemma:adding-deleting-link-beneficial-in-SV-SRN-MO-Framework}, and Definition \ref{def:stability-storage-constraints} of bilateral stability.	\qed
\end{proof}

Now, we look at the stability point of SVN and SV-SRN under the MO-Framework.


\begin{theorem}\label{prop:equilibrium-point}
\textcolor{mycolor}{Let $\mathfrak{g}$ be an SVN under the MO-Framework. Then, the stability point $\hat{\eta}$ of $\mathfrak{g}$ is unique and is given by
$\hat{\eta} = \ceil*{\frac{|\ln(\frac{c}{\beta(1-\lambda)})|}{|\ln\lambda|}} = \floor*{\frac{|(\ln\frac{c\lambda}{\beta(1-\lambda)})|}{|\ln\lambda|}}$.}
\end{theorem}
\begin{proof}
From  Lemma \ref{lemma:add}, adding a link for agent $i$ is beneficial if and only if 


$\eta_{i}(\mathfrak{g}) \ln \lambda>\ln(\frac{c}{\beta (1-\lambda)})$, if and only if

$\eta_{i}(\mathfrak{g}) <\frac{|\ln(\frac{c}{\beta (1-\lambda)})|}{|\ln \lambda|}$\\

\noindent {Hence,} for agent $i$, increasing neighborhood size is not beneficial if and only if

 $\eta_{i}(\mathfrak{g}) \geq \frac{|\ln(\frac{c}{\beta (1-\lambda)})|}{|\ln \lambda|}$.\\


\noindent Similarly, from Lemma \ref{lemma:delete}, deleting a link for agent $i$ is beneficial if and only if 



$\ln(\frac{c\lambda}{\beta(1-\lambda)})>{\eta_{i}(\mathfrak{g})} \ln \lambda$, if and only if

$\frac{|\ln(\frac{c\lambda}{\beta(1-\lambda)})|}{|\ln \lambda|}<{\eta_{i}(\mathfrak{g})}$\\

\noindent {So,} decreasing neighborhood size is not beneficial for agent $i$ if and only if
$\frac{|\ln(\frac{c\lambda}{\beta(1-\lambda)})|}{|\ln \lambda|}\geq{\eta_{i}(\mathfrak{g})}$.\\


\noindent {Therefore,} $L = \frac{|\ln(\frac{c}{\beta(1-\lambda)})|}{|\ln\lambda|}$ and $U = \frac{|\ln(\frac{c\lambda}{\beta(1-\lambda)})|}{|\ln\lambda|}$ are, respectively, the lower and upper bounds of $\hat{\eta}$. \\

\noindent $\text{U }  = \frac{|\ln(\frac{c}{\beta(1-\lambda)})|}{|\ln\lambda|} + \frac{|\ln\lambda|}{|\ln\lambda|}$
$= L + 1$. \\

\noindent It is easy to see that if $L$ is not an integer (and hence, $U$ is not an integer), the stability point $\hat{\eta}$ is the unique positive integer between $L$ and $U$. 
\qed
\end{proof}

\begin{remark}
\textcolor{mycolor}{For most values of $c, \beta,$ and $\lambda$, $\frac{|\ln(\frac{c}{\beta(1-\lambda)})|}{|\ln\lambda|}$, and hence, $\frac{|\ln(\frac{c\lambda}{\beta(1-\lambda)})|}{|\ln\lambda|}$ are non-integers.}
\end{remark}
\begin{example}\label{exmp:stability-point}
\sloppy
Consider the networks $\mathfrak{g}$ and $\mathfrak{s}$ (see Fig. \ref{fig:stable-example-network-s-g-observation}). In both the networks, let the cost $c=0.0055$, $\beta=0.6$, and $\lambda=0.2$. Here, $\ceil*{\frac{|\ln(\frac{c}{\beta(1-\lambda)})|}{|\ln\lambda|}}=\ceil*{2.72}$ and $\floor*{\frac{|(\ln\frac{c\lambda}{\beta(1-\lambda)})|}{|\ln\lambda|}}=\floor*{3.72}$, and hence, $\hat{\eta}=3$. In network $\mathfrak{g}$, all agents have three neighbours each, and hence, $\mathfrak{g}$ is bilaterally stable. Despite the fact that agent $g$ in the network $\mathfrak{s}$ has an incentive to add one more link, network $\mathfrak{s}$ is also bilaterally stable.

\begin{figure*}[!h]
\centering
\subfigure[Network $\mathfrak{g}$ ]
{
\includegraphics[scale=0.11]{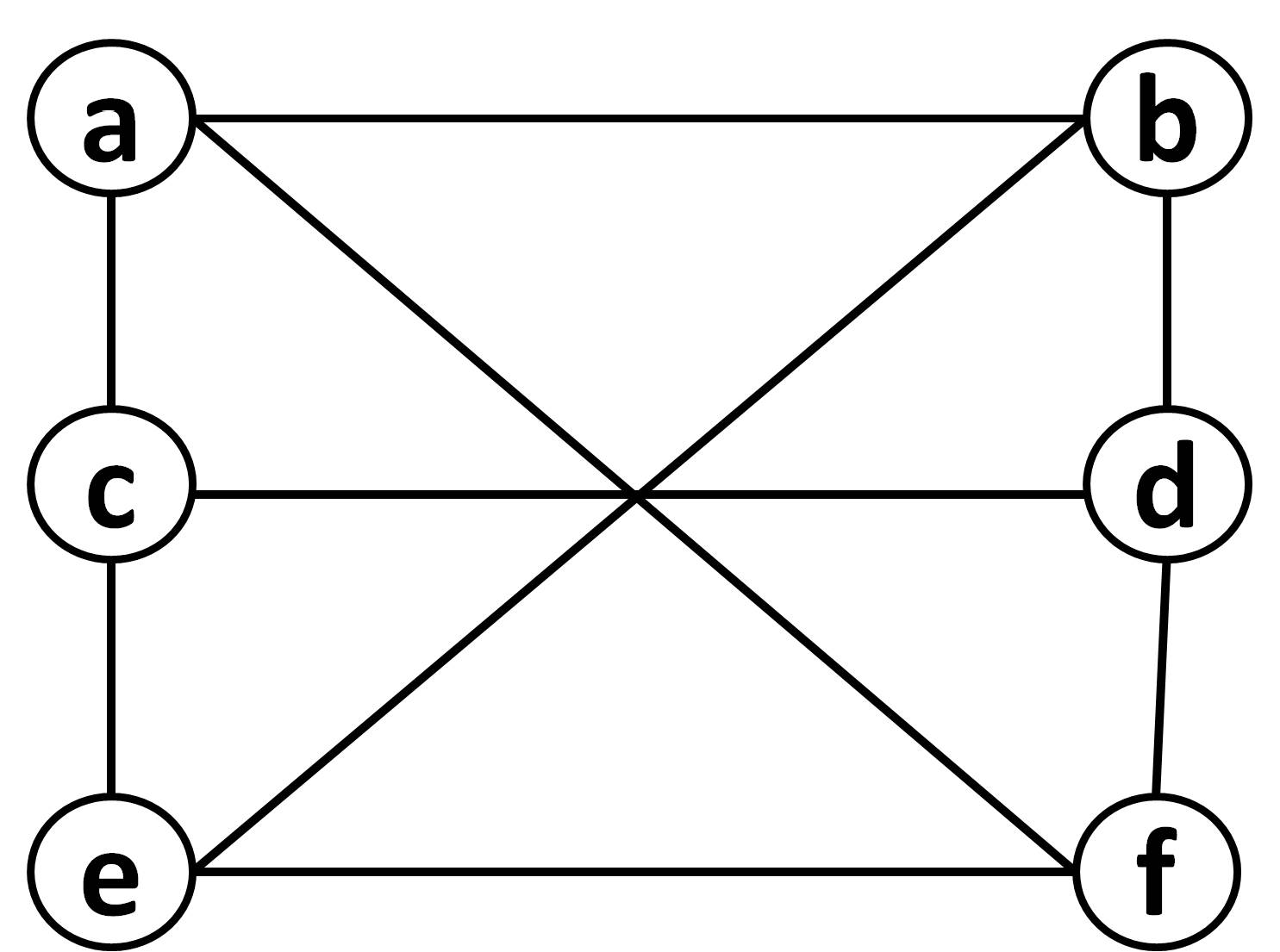} 
\label{fig:stable-example-7-agents-g-network}
}
\quad \quad \quad
\subfigure[Network $\mathfrak{s}$]
{
\includegraphics[scale=0.11]{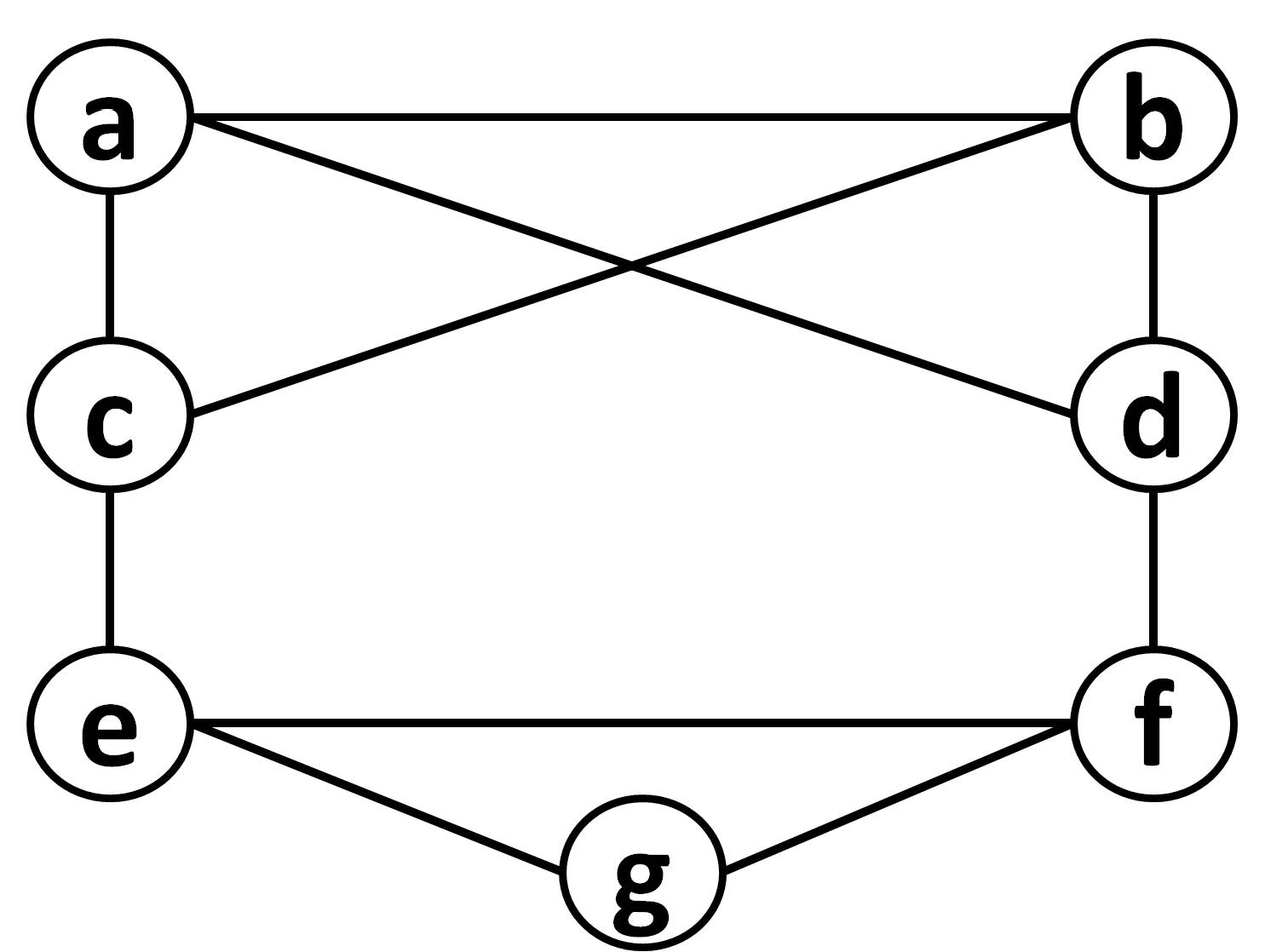} 
\label{fig:stable-example-7-agents-s-network}
}
\caption{Stable SVN Networks under the MO-Framework with Sufficient Storage}
\label{fig:stable-example-network-s-g-observation}
\end{figure*}
\end{example}

Now, we derive the stability point for SV-SRN network under the MO-Framework. Here, Definition \ref{def:stability-storage-constraints} is relevant, and for simplicity, we assume that $\frac{s}{d}$ is an integer.

\begin{theorem}\label{theorem:SV-SRN-Stability-Point}
Let $\mathfrak{g}$ be an SV-SRN, under the MO-Framework. 
 
Then, $\tilde{n}=min\{\hat{\eta}, \frac{s}{d} \}$, is the unique stability point of $\mathfrak{g}$.  
\end{theorem}
\begin{proof}
If all agents have sufficient storage, then from Theorem \ref{prop:equilibrium-point}, $\hat{\eta}$ is the stability point. 


Now, let us assume that each agent has a total amount of storage, $s$, available for sharing, $d$ amount of data to backup. Then, $\frac{s}{d}$ defines the maximum possible neighborhood size of each agent in the network.

Therefore, $min\{\hat{\eta}, \frac{s}{d} \}$ is the stability point, given $c, \lambda, \beta, s, d$. 

Alternatively, we may also use the bound $\frac{s}{d}$ in Lemmas \ref{lemma:add}, \ref{lemma:delete} and Theorem \ref{prop:equilibrium-point}. \qed
\end{proof} 

Henceforth, for the sake of uniformity, we shall use $\hat{\eta}$ (and not $\tilde{n}$)
for the stability point of SV-SRN under the MO-Framework too.

\subsection{Characterization Under the SO-Framework}
In this subsection, we derive necessary and sufficient conditions for bilateral stability of SVN and SRN under the SO-Framework, and then discuss the stability point of these networks.

\begin{lemma}\label{lemma:susb-always-add}
In an SVN, $\mathfrak{g}$, under the SO-Framework, for any agent $i\in \mathfrak{g}$, forming a partnership with another agent $j\in \mathfrak{g}$ is always beneficial. 
\end{lemma}
\begin{proof}

As $\mathfrak{g}$ is an SVN, under the SO-Framework,
$u_{i}(\mathfrak{g})=\beta(1-\lambda^{\eta_{i}(\mathfrak{g})})$, for all $i \in \mathfrak{g}$.\\
For $i,j \in \mathfrak{g}$, if $\langle ij \rangle \not \in \mathfrak{g}$, then
 $u_{i}(\mathfrak{g}+\langle ij\rangle)=[\beta(1-\lambda^{\eta_{i}(\mathfrak{g})+1})]$.\\


\noindent Adding a new link or backup partner is beneficial for agent $i$ if and only if

$u_{i}(\mathfrak{g}+\langle ij \rangle)> u_{i}(\mathfrak{g})$, if and only if

$[\beta(1-\lambda^{\eta_{i}(\mathfrak{g})+1})] > [\beta(1-\lambda^{\eta_{i}(\mathfrak{g})})]$, if and only if

$ \lambda^{\eta_{i}(\mathfrak{g})+1} < \lambda^{\eta_{i}(\mathfrak{g})}$, if and only if	

$ \lambda < 1$, which is always true.	\qed
\end{proof}

\begin{corollary}\label{lemma:susb-never-del}
In an SVN, $\mathfrak{g}$, under the SO-Framework, no agent benefits by deleting any existing partnership. 
\end{corollary}

\begin{theorem}\label{coro:SVN-SO-Framework-Stability}
An SVN, $\mathfrak{g}$, under the SO-Framework, is bilaterally stable if and only if $\mathfrak{g}$ is a complete network.
\end{theorem}
\begin{proof}
Follows from Lemma \ref{lemma:susb-always-add} and Corollary \ref{lemma:susb-never-del}.
\end{proof}

Now, we state and prove the following for SRN, under the SO-Framework.

\begin{lemma}\label{lemma:lemma:adding-deleting-link-beneficial-in-SRN-SO-Framework}
In an SRN, $\mathfrak{g}$, under the SO-Framework, for any agent $i\in \mathfrak{g}$, forming a partnership with another agent $j\in \mathfrak{g}$ is beneficial if and only if $b-c\eta_{i}(\mathfrak{g})\geq c$ and $s-d\eta_{j}(\mathfrak{g})\geq d$.
\end{lemma}
\begin{proof}
In the SO-Framework, the utility of each agent $i \in \mathfrak{g}$ increases with increase in its neighborhood size $\eta_i(\mathfrak{g})$. \\

Therefore, 
for any agent $i\in \mathfrak{g}$, forming a partnership with another agent $j\in \mathfrak{g}$ is beneficial if and only if agent $i$'s budget allows this link addition and agent $j$ has free storage space for agent $i$'s data. (Refer Definition \ref{def:stability-storage-budget-constraints}). \\

Agent $j$ has free storage space for $i$'s data, if and only if $s-d\eta_{j}(\mathfrak{g})\geq d$. (Similar to the proof of Lemma \ref{lemma:adding-deleting-link-beneficial-in-SV-SRN-MO-Framework}). \\ 

Similarly, agent $i$'s budget allows adding a link, if and only if $ b-c\eta_{i}(\mathfrak{g})\geq c$. 
\qed
\end{proof}

\begin{corollary}\label{lemma:deleting-link-beneficial-in-SRN-SO-Framework}
In an SRN, $\mathfrak{g}$, under the SO-Framework, no agent benefits by deleting any existing partnership. 
\end{corollary}

\begin{theorem}\label{theorem:stability-SV-SRN-SO-Framework}
An SRN $\mathfrak{g}$ under the SO-Framework is bilaterally stable if and only if
$[b-c\eta_{i}(\mathfrak{g})\geq c$ and $s-d\eta_{j}(\mathfrak{g})\geq d]  \Rightarrow [b-c\eta_{j}(\mathfrak{g})<c$ or $s-d\eta_{i}(\mathfrak{g})<d]$, for all $\langle ij\rangle \not \in \mathfrak{g}$.
\end{theorem}
\begin{proof}
Follows from Lemma \ref{lemma:lemma:adding-deleting-link-beneficial-in-SRN-SO-Framework}, and Definition \ref{def:stability-storage-budget-constraints}. \qed
\end{proof}

\begin{remark}
Since in an SRN $\mathfrak{g}$ under the SO-Framework, no agent benefits by deleting any existing partnership, link deletion does not appear in the bilateral stability conditions above.
\end{remark}



Now, we look at the stability point of SVN and SRN under the SO-Framework. The following case (Theorem \ref{lemma:stability-point-susb-network}) is the only case where the stability point depends on the number of agents, $N$. 
\begin{theorem}\label{lemma:stability-point-susb-network}
In an SVN $\mathfrak{g}$, under the SO-Framework, $\hat{\eta} = N - 1$, is the unique stability point, where $N$ is the number of agents.
\end{theorem}
\begin{proof}
Follows from Lemma \ref{lemma:susb-always-add}. 
\qed
\end{proof}

Except SVN under the SO-Framework, in all other scenarios (including the following), the stability point is independent of $N$. In all cases (including the above), the stability point is unique. In the following, for simplicity, we assume that $\frac{s}{d}$ and $\frac{b}{c}$ are integers.

\begin{theorem}\label{lemma:stability-point-sb-network}
In an SRN $\mathfrak{g}$, under the SO-Framework, 
$\hat{\eta}=min\{\frac{s}{d},\frac{b}{c}\}$, for all $i\in \mathfrak{g}$, is the unique stability point, where no agent has incentive to add or delete a link. 
\end{theorem}
\begin{proof}
A constructive proof follows from Lemma \ref{lemma:lemma:adding-deleting-link-beneficial-in-SRN-SO-Framework}. \\
Alternatively, it is clear that it is beneficial for each agent to add as many links as possible. The degree of agent $i$ in $\mathfrak{g}$, $\eta_{i}(\mathfrak{g})$, is limited only by its storage space $s$ and budget $b$. That is, 
 
$ s \geq d \eta_{i}(\mathfrak{g})$ and 
$ b \geq c \eta_{i}(\mathfrak{g})$. \\
The theorem follows as the above is true for all $i \in \mathfrak{g}$. \qed
\end{proof}

\begin{example}
Let us consider the networks $\mathfrak{g}$ (see Fig. \ref{fig:s-b-cnstraints-6-agents-g-network}) and $\mathfrak{s}$ (see Fig. \ref{fig:s-b-cnstraints-6-agents-s-network}), each consisting of six agents, and network $\mathfrak{t}$ (see Fig. \ref{fig:s-b-cnstraints-7-agents-t-network}) consisting of seven agents. Assume that, in $\mathfrak{g}$ and $\mathfrak{t}$,  $s=60$ TB, $d=20$ TB, $b=0.5$, and $c=0.1$. Assume that, in network $\mathfrak{s}$, $s=60$ TB, $d=10$ TB, $b=0.4$, and $c=0.1$.
\begin{figure*}[!h]
\centering
\subfigure[Network $\mathfrak{g}$ ]
{
\includegraphics[scale=0.11]{stable-example-6-players.jpg} 
\label{fig:s-b-cnstraints-6-agents-g-network}
}
\quad \quad \quad
\subfigure[Network $\mathfrak{s}$ ]
{
\includegraphics[scale=0.11]{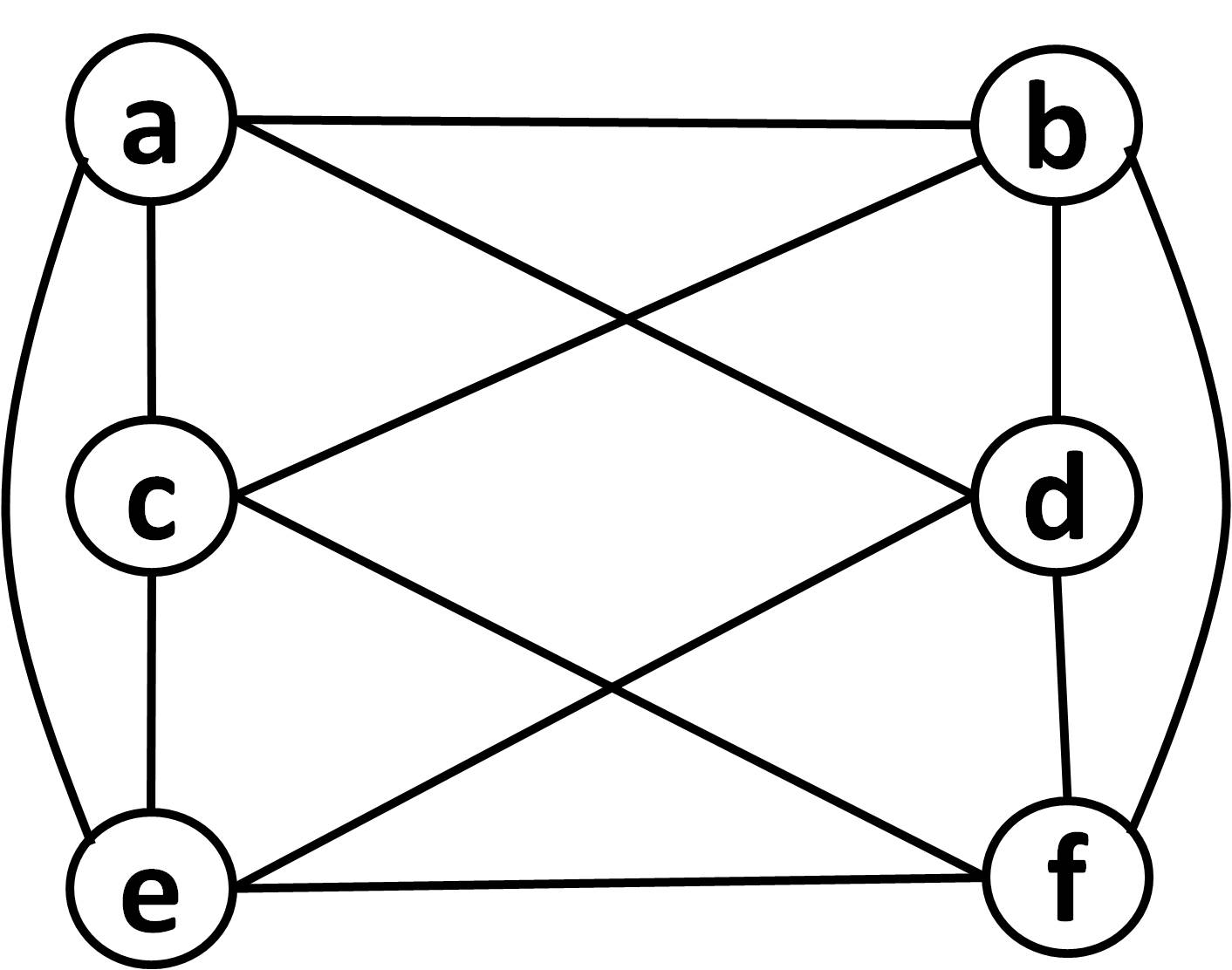} 
\label{fig:s-b-cnstraints-6-agents-s-network}
}
\quad \quad \quad
\subfigure[Network $\mathfrak{t}$]
{
\includegraphics[scale=0.11]{stable-example-7-players.jpg} 
\label{fig:s-b-cnstraints-7-agents-t-network}
}
\caption{Stable SRN Networks under the SO-Framework}
\label{fig:s-b-cnstraints-observation}
\end{figure*}

Note that in networks $\mathfrak{g}$ and $\mathfrak{t}$, although the budget constraints permit agents to maintain five neighbors each, storage limitations do not permit agents to maintain more than three neighbors each, and hence, $\mathfrak{g}$ and $\mathfrak{t}$ are bilaterally stable networks.

In network $\mathfrak{s}$, although storage constraints permit agents to maintain six neighbors each, budget constraints do not allow agents to maintain more than four neighbors each. Hence, $\mathfrak{s}$ is bilaterally stable. \qed
\end{example}


We summarize the above results on bilateral stability conditions and stability point in Table \ref{table:stability-summary} and Table \ref{table:stability-point}, respectively.
\begin{table}[h!]
\small
\centering
\caption{Summary of stability condition for different network.}
\begin{tabular}{p{1.7cm} p{1.9cm} p{7cm}}
\hline
\textbf{Network Type} & \textbf{Framework} & \textbf{Condition(s) for Bilateral Stability} \\ 
\hline
&&\\
SVN & MO-Framework& 
1. For all $\langle ij\rangle \in \mathfrak{g}$,  

$\mbox{        }\quad\quad\textcolor{mycolor}{\beta \lambda^{\eta_{i}(\mathfrak{g})-1}[1-\lambda]< c \Rightarrow \beta \lambda^{\eta_{j}(\mathfrak{g})-1}[1-\lambda]>c}$, and

2. For all $\langle ij\rangle \not \in \mathfrak{g}$, 

$\mbox{        }\quad\quad\textcolor{mycolor}{\beta \lambda^{\eta_{i}(\mathfrak{g})}[1-\lambda]>c \Rightarrow \beta \lambda^{\eta_{j}(\mathfrak{g})}[1-\lambda]<c}$.\\ 

&&\\

SV-SRN & MO-Framework &
1. For all $\langle ij\rangle \in \mathfrak{g}$, 

$\quad\quad\textcolor{mycolor}{\beta \lambda^{\eta_{i}(\mathfrak{g})-1}[1-\lambda]< c \Rightarrow \beta \lambda^{\eta_{j}(\mathfrak{g})-1}[1-\lambda]>c}$, and

2. For all $\langle ij\rangle \not \in \mathfrak{g}$, 

$\quad\quad\textcolor{mycolor}{\beta \lambda^{\eta_{i}(\mathfrak{g})}[1-\lambda] >c}$ and $s-d\eta_{j}(\mathfrak{g})\geq d$

$\quad\quad\quad\quad\quad\quad \Rightarrow \textcolor{mycolor}{\beta \lambda^{\eta_{j}(\mathfrak{g})}[1-\lambda]<c}$ or $s-d\eta_{i}(\mathfrak{g})< d$.\\

&&\\

SVN & SO-Framework &  Each agent $i\in \mathfrak{g}$ has backup partnerships with all agents $j \in \mathfrak{g}$ with $j\not=i$.\\

&&\\

SRN & SO-Framework & For all $\langle ij\rangle \not \in \mathfrak{g}$,  

$\quad\quad[b-c\eta_{i}(\mathfrak{g})\geq c$ and $s-d\eta_{j}(\mathfrak{g})\geq d]$ 

$\quad\quad\quad\quad\quad\quad \Rightarrow [b-c\eta_{j}(\mathfrak{g})<c$ or $s-d\eta_{i}(\mathfrak{g})< d]$.\\
\hline
\end{tabular}

\label{table:stability-summary}
\end{table} 
\begin{table}[h!]
\small
\centering
\caption{Summary of stability point for different network types under MO- and SO-frameworks.}
\begin{tabular}{p{1.8cm} p{2cm} p{4.05cm}}
\hline
\textbf{Network Type} &  \textbf{Framework} & \textbf{Unique Stability Point} \\ 
\hline
&&\\
SVN & MO-Framework & $\hat{\eta} = \ceil*{\frac{|\ln(\frac{c}{\beta(1-\lambda)})|}{|\ln\lambda|}} = \floor*{\frac{|(\ln\frac{c\lambda}{\beta(1-\lambda)})|}{|\ln\lambda|}}$.\\ 
&&\\
SV-SRN & MO-Framework & $\tilde{n}=min\{\hat{\eta}, \frac{s}{d} \}$ \\
&&\\
SVN & SO-Framework & $\hat{\eta}=N-1$\\
&&\\
SRN & SO-Framework & $\hat{\eta}=min\{ \frac{s}{d},  \frac{b}{c} \}$\\
&&\\
\hline
\end{tabular}
\label{table:stability-point}
\end{table}
\section{Stable, Efficient and Contented Networks}\label{subsec:stabile-networks}
We first discuss conditions on $N$ and $\hat{\eta}$ for connected networks to be bilaterally stable in Section \ref{subsubsec:single}. We, then, look at networks that are comprised of multiple connected components, and discuss conditions on $N$, $\hat{\eta}$ as well as number of agents in individual components that lead to a bilaterally stable network in Section \ref{subsubsec:components}. Finally, we discuss conditions that lead to unique bilaterally stable networks in Section \ref{subsubsec:unique}.

Henceforth, whenever we say $\mathfrak{g}$ is a symmetric social storage network, $\mathfrak{g}$ may be any of the networks SVN, SRN or SV-SRN with $N$ agents, under the MO- or SO- Framework,  with the unique stability point $\hat{\eta}$ corresponding to that network type and framework. 

\subsection{Stable Networks}
\textcolor{mycolor}{Up to this point, we have not explicitly discussed the process of network formation. This is because all our results above are independent of any process or protocol for network formation. However, the following results depend on the where we start the network formation from (refer \citep{Dai-Network-Formation} for different network configurations). We consider networks that evolve either from a null network (where all agents are initially disconnected) or from a complete network (where all agents are initially connected). When a network evolves from the null network, every agent starts contacting other agents to form links, in no particular order. This happens until there is no pair of agents who would consent to form a link. Similarly, when a network evolves from the complete network, pairs of agents consider deleting links if beneficial.}





\subsubsection{Connected Stable Networks}\label{subsubsec:single}
We start our discussion with the following remark. 

\begin{remark}
Each agent {\em{aims}} to achieve neighborhood size $\hat{\eta}$.
\end{remark}

Though agents want to achieve neighborhood size $\hat{\eta}$, this may not always be possible. 
The following example demonstrates how stable networks may evolve when all agents are isolated (Fig. \ref{fig:delete-network-1less}) or all connected (Fig. \ref{fig:delete-network-2more}), initially. 


\begin{example}\label{exmp:achieve-stability}

\begin{figure*}[!h]
\centering
\subfigure[Network $\mathfrak{g}$ ]
{
\includegraphics[scale=0.11]{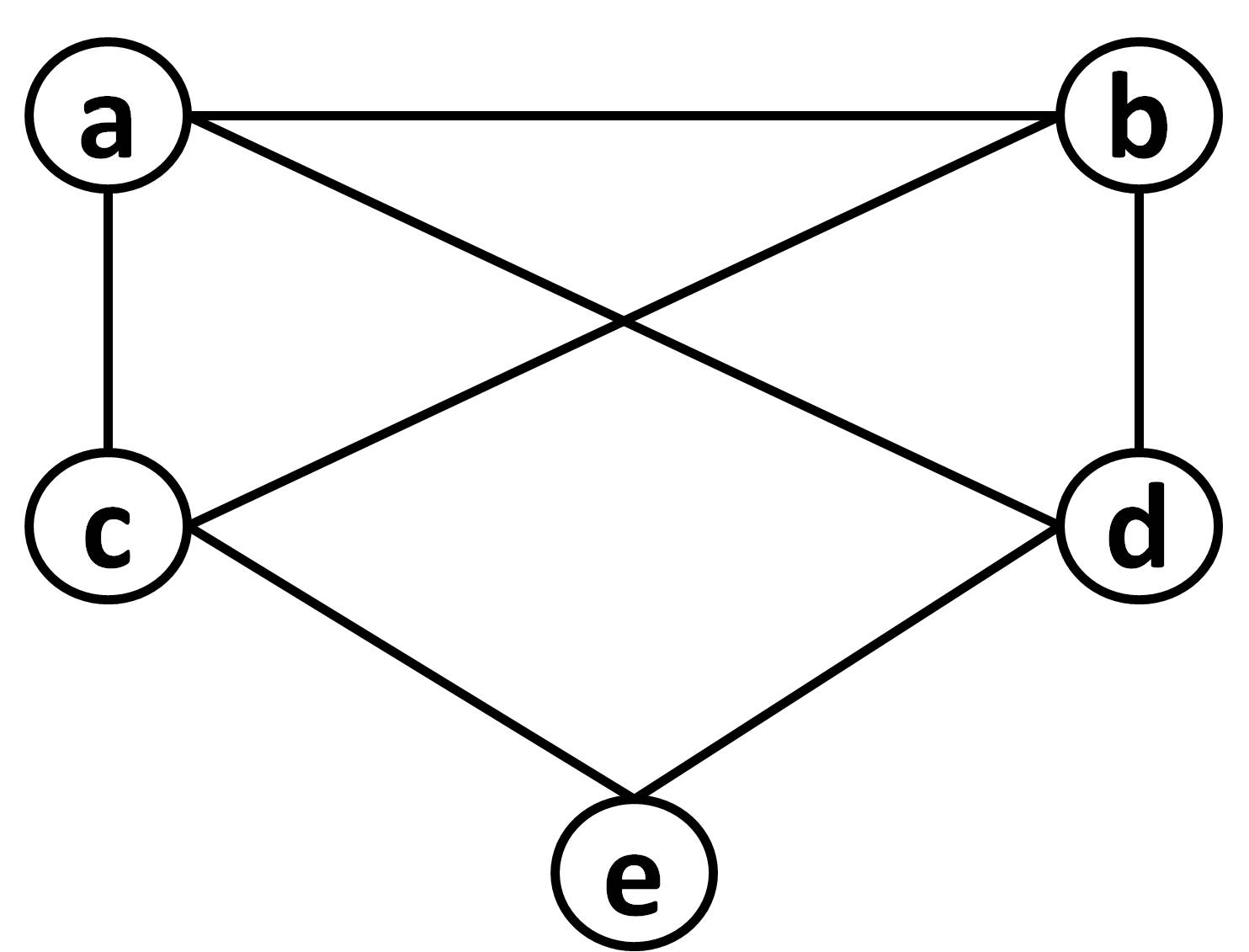} 
\label{fig:delete-network-1less}
}
\quad \quad \quad
\subfigure[Network $\mathfrak{s}$ ]
{
\includegraphics[scale=0.11]{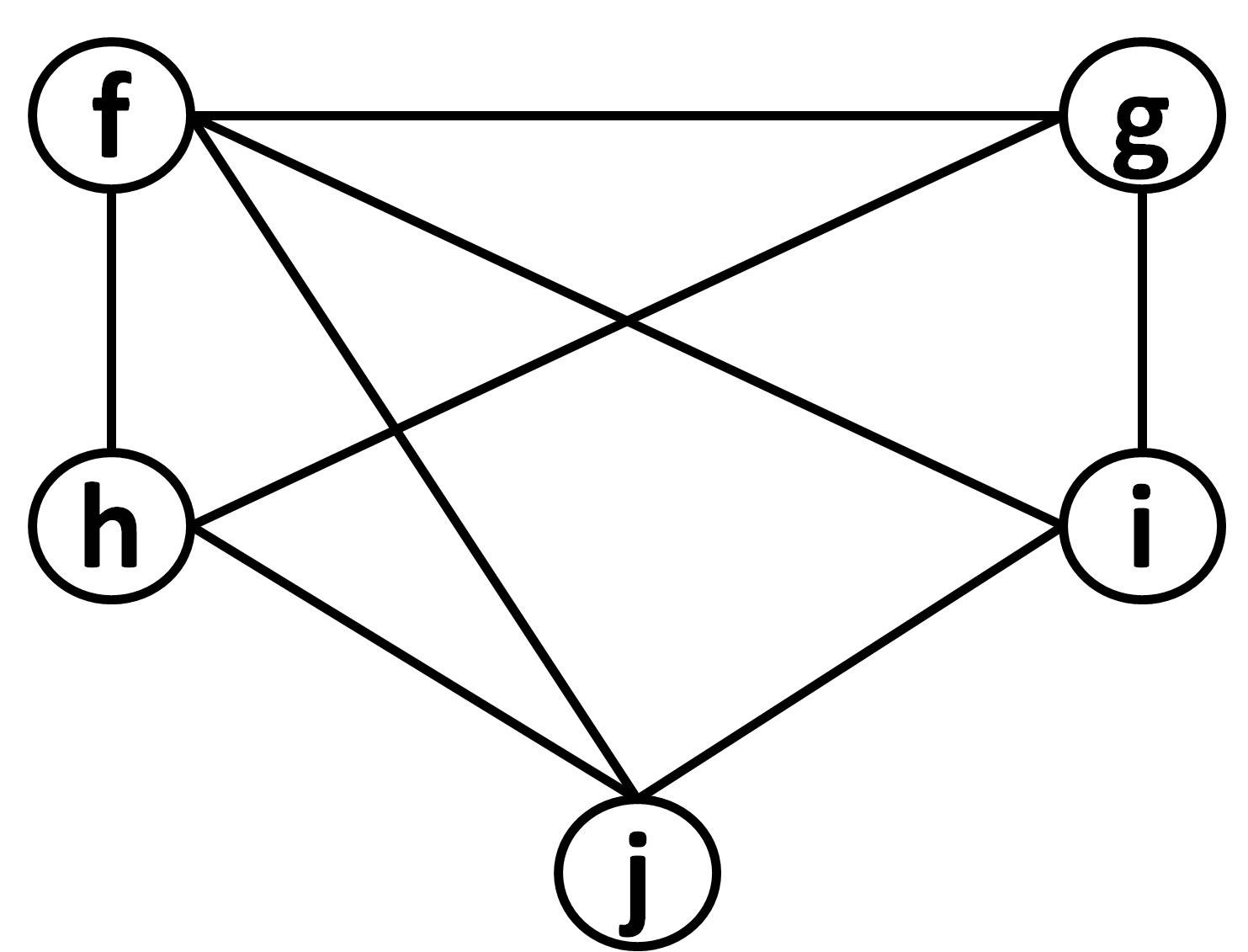} 
\label{fig:delete-network-2more}
}

\caption{Stable Networks $\mathfrak{g}$ evolved from the Null Network, and $\mathfrak{s}$ from the Complete Network}
\label{fig:s-b-delete-observation-observation}
\end{figure*}

Consider networks $\mathfrak{g}$ and $\mathfrak{s}$. Let $\hat{\eta} = 3$. Networks $\mathfrak{g}$ and $\mathfrak{s}$ are both bilaterally stable, where $\mathfrak{g}$ evolves from the empty network and $\mathfrak{s}$ evolves from the complete network. In $\mathfrak{g}$, although agent $e$ has an incentive to add another link, no other agent (who does not have a link with $e$) would consent to adding a new link with $e$ as they (that is, agents $a$, $b$, $c$ and $d$) have already reached their stability point $\hat{\eta}$ (that is, their neighborhood size is $\hat{\eta} = 3$) and hence, have no incentive to add or delete any link. 

In $\mathfrak{s}$, although agent $f$ has an incentive to delete a link, no other agent (who has a link with $f$) would consent to deleting their link with $f$ as they (that is, agents $g$, $h$, $i$ and $j$) have already reached their stability point $\hat{\eta}$ (that is, their neighborhood size is $\hat{\eta} = 3$) and hence, have no incentive to add or delete any link. 
\qed
\end{example}

In Propositions \ref{lem:n-1-agent} and \ref{lem:g-is-even-stable} below, we provide results that would be useful for an independent observer in checking for a pairwise stable symmetric social storage network, how many agents have maximised their utility. Thus, as discussed in the introduction, such an observer (say, an administrator or regulator) can externally perturb the system so that all agents achieve maximum utility. 

\begin{proposition}\label{lem:n-1-agent}
Let $N$ and $\hat{\eta}$ be (positive) odd integers, with $\hat{\eta}<N$. Then:
\begin{enumerate}
\item  Any symmetric social storage network $\mathfrak{g}$ with $N$ agents and stability point $\hat{\eta}$ consists of at least one agent who has an incentive to either add or delete a link.
\item There exists a connected, bilaterally stable, symmetric social storage network with exactly $N-1$ agents who have no incentive to add or delete any link.
\end{enumerate}
\end{proposition}
\begin{proof}
Let $\mathfrak{g}$ be bilaterally stable, and let $\ell$ be the number of links in $\mathfrak{g}$. 

$\hat{\eta} < N$ ensures that $\ell$ does not exceed the maximum number of links $\mathfrak{g}$ can possibly have, that is, $\frac{N\times (N-1)}{2}$. 

As the utility of each agent is maximum when its neighborhood size is $\hat{\eta}$, total number of links $\tilde\ell = \frac{N\times \hat{\eta}}{2}$ will be attained if possible. 
However, $\tilde{\ell}$ is not an integer, as both $N$ and $\hat{\eta}$ are odd. 

This implies that, not all $N$ agents have a neighborhood size of $\hat{\eta}$ at stability. This proves (1).

Now, $N-1$ agents having $\hat{\eta}$ neighbors and the $N^{th}$ agent having $\hat{\eta}-1$ or $\hat{\eta}+1$ neighbors are, however, possible. Let g be such a network with exactly $N-1$ agents who have no incentive to add or delete any link. These $N-1$ agents have neighborhood size $\hat{\eta}$. None of these $N-1$ agents will consent to add or delete any link
(among themselves, or with the $N^{th}$ agent). Thus, the symmetric social storage network $\mathfrak{g}$ is bilaterally stable. If $\mathfrak{g}$ is connected, we are done. Otherwise, all non-trivial components (that is, components with 2 or more agents in each) of $\mathfrak{g}$ can be connected as follows, without changing the neighborhood sizes of any of the agents. Let
$\langle i_{1} j_{1} \rangle$ and $\langle i_{2} j_{2} \rangle$ be links in two different (non-trivial) components, say $\mathfrak{g}(\kappa_{1})$ and $\mathfrak{g}(\kappa_{2})$ of $\mathfrak{g}$. Deleting both these links, and replacing them with $\langle i_{1} j_{2} \rangle$ and $\langle i_{2} j_{1} \rangle$ connects $\mathfrak{g}(\kappa_{1})$ and $\mathfrak{g}(\kappa_{2})$, without changing the neighborhood sizes of any of the agents. As neighborhood sizes of all the agents remain the same, the resulting graph is bilaterally stable too. Now, if $i$ is an isolated agent and $\langle jk \rangle$ is a link in $\mathfrak{g}$, delete 
$\langle jk \rangle$, and add $\langle ij \rangle$ and $\langle ik \rangle$ instead. In this case, clearly, $i$ continues to be the only agent with an incentive to either add or delete a link. This proves (2).	\qed
\end{proof}

\begin{remark}
In the proof of Proposition \ref{lem:n-1-agent}, on the one hand, when the network evolves from the null network, $\hat{\eta}-1$ neighbors for the $N^{th}$ agent is as beneficial as possible, and the total number of links will, hence, be ${{\ell}}=\frac{[(N-1) \hat{\eta} + (\hat{\eta}-1)]}{2}$. 

On the other hand, when the network evolves from the complete network, $\hat{\eta}+1$ neighbors for the $N^{th}$ agent is as beneficial as possible, and the total number of links will, hence, be ${{\ell}}=\frac{[(N-1) \hat{\eta} + (\hat{\eta}+1)]}{2}$. This number also does not exceed the maximum possible number of links, as $\hat{\eta} \leq N-2$ (because $\hat{\eta} < N$, and both $\hat{\eta}$ and $N$ are odd). 
\end{remark}

\begin{proposition}\label{lem:g-is-even-stable}
Let at least one of $N$ and $\hat{\eta}$ be even, and let $\hat{\eta}<N$. Then, there exists a connected bilaterally stable symmetric social storage network $\mathfrak{g}$ where no agent has incentives to add or delete any link.
\end{proposition}
\begin{proof}
\textcolor{mycolor}{Existence of the $\hat{\eta}$-regular network on $N$ agents, $\mathfrak{g}$, follows trivially from the Erd\H{o}s--Gallai theorem. Clearly $\mathfrak{g}$ is a bilaterally stable.}
\end{proof}


\subsubsection{Stable Network with Multiple Connected Components}\label{subsubsec:components}
We, now, discuss results on stability of symmetric storage networks with two or more components. 
Examples of scenarios where this might be useful include companies under the same umbrella group, where the social storage networks of each of these companies may be viewed as 
a component of a larger network, which may be monitored or analysed by an independent observer (as discussed in the previous section). 

\begin{clm}\label{atmost-one-with-lessthan-eta}
Suppose $\mathfrak{g}$ is a symmetric social storage network with two or more components. If $\mathfrak{g}$ is bilaterally stable, then there is at most one component with less than or equal to $\hat{\eta}$ agents.
\end{clm}

\begin{proof}

\textcolor{mycolor}{Suppose, $\mathfrak{g}(\kappa_i)$ and $\mathfrak{g}(\kappa_j)$ are two different (non-empty) components with less than or equal to $\hat{\eta}$ agents. It is easy to see that all agents in $\mathfrak{g}(\kappa_i)$ as well as in $\mathfrak{g}(\kappa_j)$ have less than $\hat{\eta}$ neighbours. Consider agents $i\in \mathfrak{g}(\kappa_i)$ and $j\in \mathfrak{g}(\kappa_j)$. Clearly, $\langle ij \rangle \not \in \mathfrak{g}$ but both $i$ and $j$ have incentives to form (at least) one link each, implying $\mathfrak{g}$ is not bilaterally stable.}
\qed
\end{proof}

\begin{proposition}\label{lem:g-is-not-stable}
Let $\mathfrak{g}$ be a symmetric social storage network which has evolved from the null network. Let $\hat{\eta}$ be odd. Suppose $\mathfrak{g}$ consists of $\kappa$ connected components, $\kappa \geq 2$. 
Suppose at least two of the components, say $\mathfrak{g}(\kappa_1)$ and $\mathfrak{g}(\kappa_2)$, each have either $\leq \hat{\eta}$ agents or an odd number of agents more than $\hat{\eta}$.  

Then $\mathfrak{g}$ is not bilaterally stable.
\end{proposition}
\begin{proof}
Follows from Proposition \ref{lem:n-1-agent} and Claim \ref{atmost-one-with-lessthan-eta}.
\qed
\end{proof}

\begin{example}
\begin{figure}[!ht]
\centering
\subfigure[Two components, $\mathfrak{g}(\kappa_{1})$ and  $\mathfrak{g}(\kappa_{2})$, of $\mathfrak{g}$]
{
\includegraphics[scale=0.13]{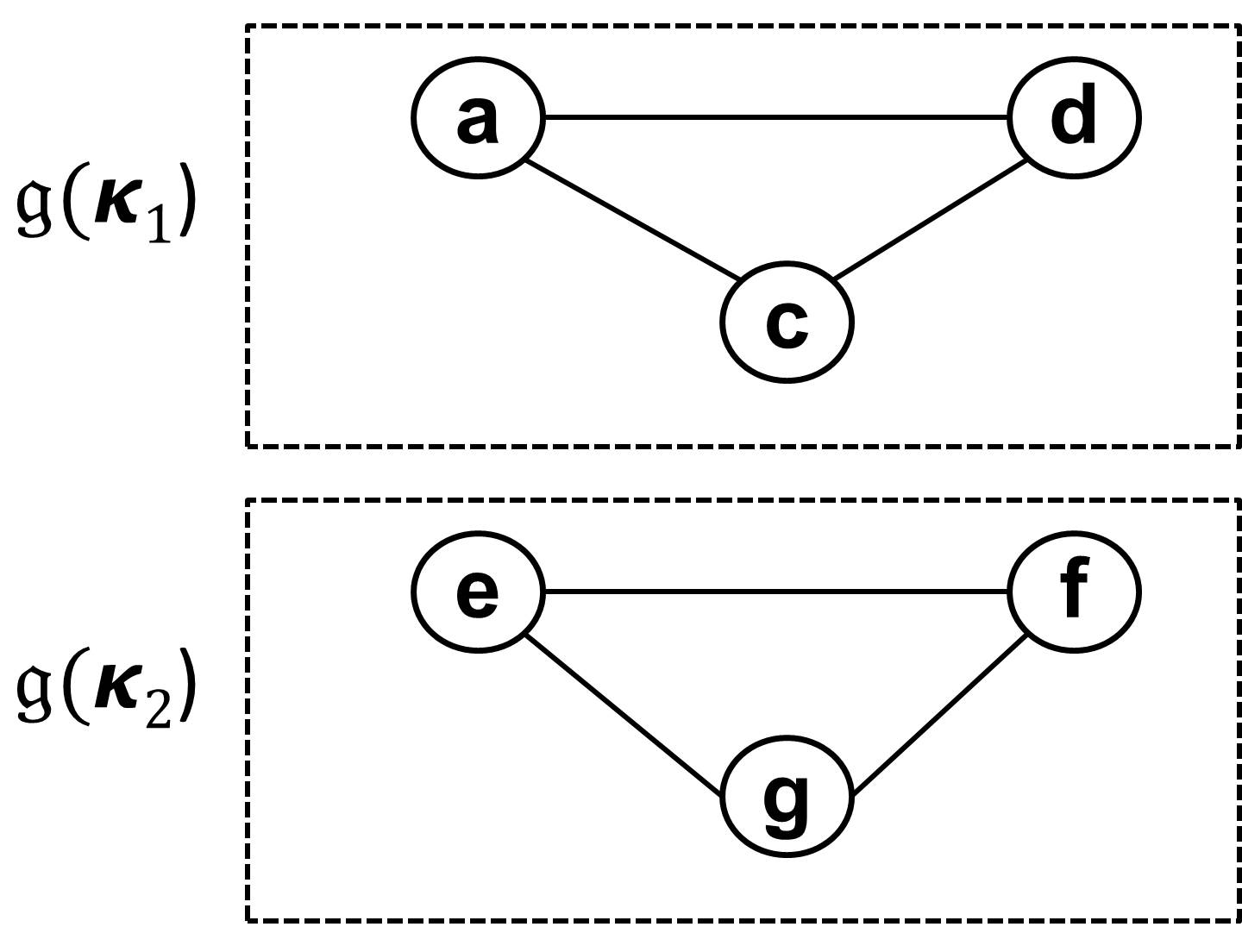} 
\label{fig:component}
}
\quad
\subfigure[Bilaterally stable network $\mathfrak{g'}$]
{
\includegraphics[scale=0.13]{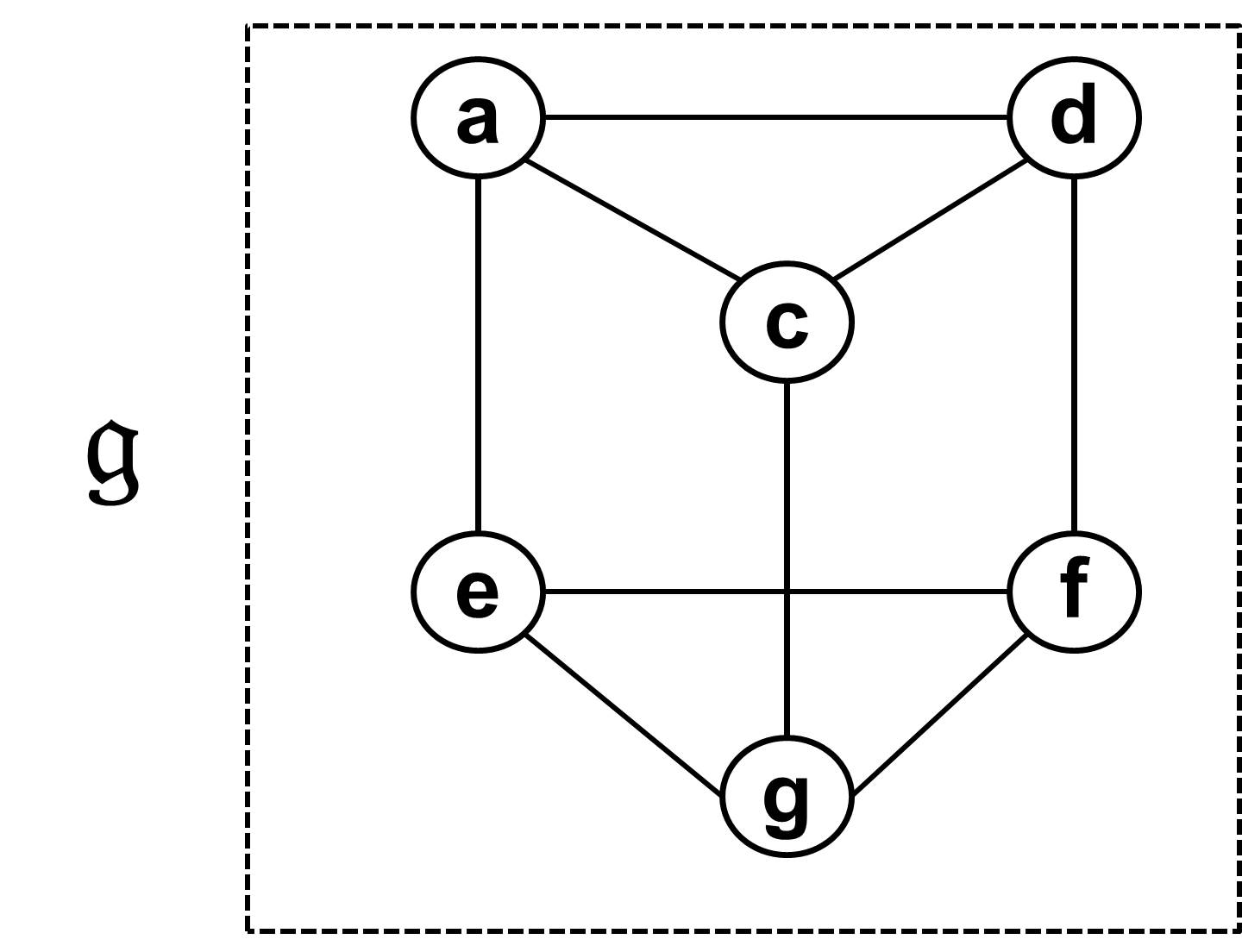} 
\label{fig:grand-network}
}
\caption{Two components $\mathfrak{g}(\kappa_{1})$ and $\mathfrak{g}(\kappa_{2})$, though complete, are bilaterally unstable when $\hat{\eta}=3$ and form a bilaterally stable network $\mathfrak{g'}$. However, the network $\mathfrak{g}$ consisting of $\mathfrak{g}(\kappa_{1})$ and $\mathfrak{g}(\kappa_{2})$ as two components is bilaterally stable when $\hat{\eta}=2$.}
\label{fig:network-component-observation}

\end{figure}

Consider network $\mathfrak{g}$ (see Fig. \ref{fig:component}), with two components $\mathfrak{g}(\kappa_{1})$ and $\mathfrak{g}(\kappa_{2})$. 

If $\hat{\eta}=3$, then both components of $\mathfrak{g}$ have $\hat{\eta}$ agents. Every agent has an incentive to add one more link. Thus $\mathfrak{g}$ is bilaterally unstable. Clearly, no agent can add any more links within the same component. The network $\mathfrak{g}'$ (see Fig. \ref{fig:grand-network}) is an example of a bilaterally stable network, which evolves from $\mathfrak{g}$.

Now, if $\hat{\eta}=2$, the network $\mathfrak{g}$ is bilaterally stable. \qed
\end{example}

\begin{corollary}\label{lem:comp-(n-1)-stable}
Let $\mathfrak{g}$ be a symmetric social storage network which has evolved from the null network and which consists of $\kappa$ components, $\kappa \geq 2$. Let $\hat{\eta}$ be odd, and let $N > \hat{\eta}$. If $\mathfrak{g}$ is bilaterally stable, then at least $\kappa-1$ components must consist of an even number of agents greater than $\hat{\eta}$.
\end{corollary}

\begin{remark}
In Proposition \ref{lem:g-is-not-stable} and Corollary \ref{lem:comp-(n-1)-stable}, if we consider networks which have evolved from the complete network, then Example \ref{fig:3odd} below acts as a counter example. 
If $\hat{\eta}$ is even, we apply Proposition \ref{lem:g-is-even-stable} to each component having more than $\hat{\eta}$ agents to see that each of these components is bilaterally stable. Now, there can be at most one component with $\leq \hat{\eta}$ agents (refer Claim \ref{atmost-one-with-lessthan-eta}), and if there is such a component, $\mathfrak{g}$ is bilaterally stable if and only if that component is complete. 
\end{remark}

The following example shows a bilaterally stable network, which has evolved from the complete network.

\begin{example}\label{eg:odd-delete-from-complete}
Let $N = 15$ and $\hat\eta = 3$. Consider the network $\mathfrak{g}$ on $N$ agents (see Fig. \ref{fig:3odd}) that consists of three components, $\mathfrak{g}(\kappa_{1})$, $\mathfrak{g}(\kappa_{2})$ and $\mathfrak{g}(\kappa_{3})$. 
Though $\mathfrak{g}$ consists of three agents, $a$, $f$ and $k$, who have an incentive to delete a link each, $\mathfrak{g}$ is bilaterally stable. This is because the agents, $a$, $f$ and $k$, are in three different components, in each of which all other agents have neighborhood size $\hat\eta$. 

\begin{figure*}[!h]
\centering
\subfigure[$\mathfrak{g}(\kappa_{1})$ ]
{
\includegraphics[scale=0.11]{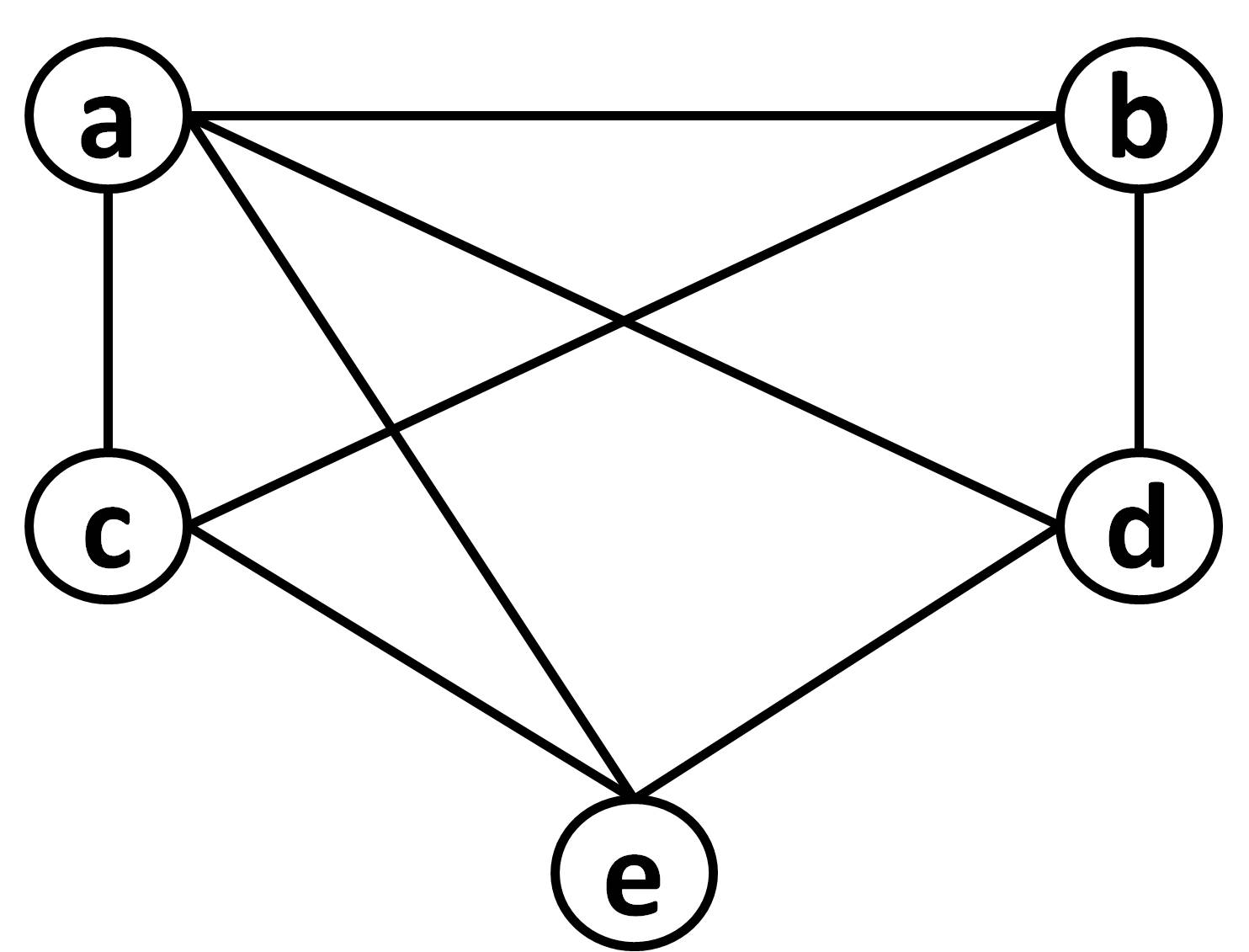} 
\label{fig:delete-network-1}
}
\quad \quad \quad
\subfigure[$\mathfrak{g}(\kappa_{2})$ ]
{
\includegraphics[scale=0.11]{delete-network-2.jpg} 
\label{fig:delete-network-2}
}
\quad \quad \quad
\subfigure[$\mathfrak{g}(\kappa_{3})$]
{
\includegraphics[scale=0.11]{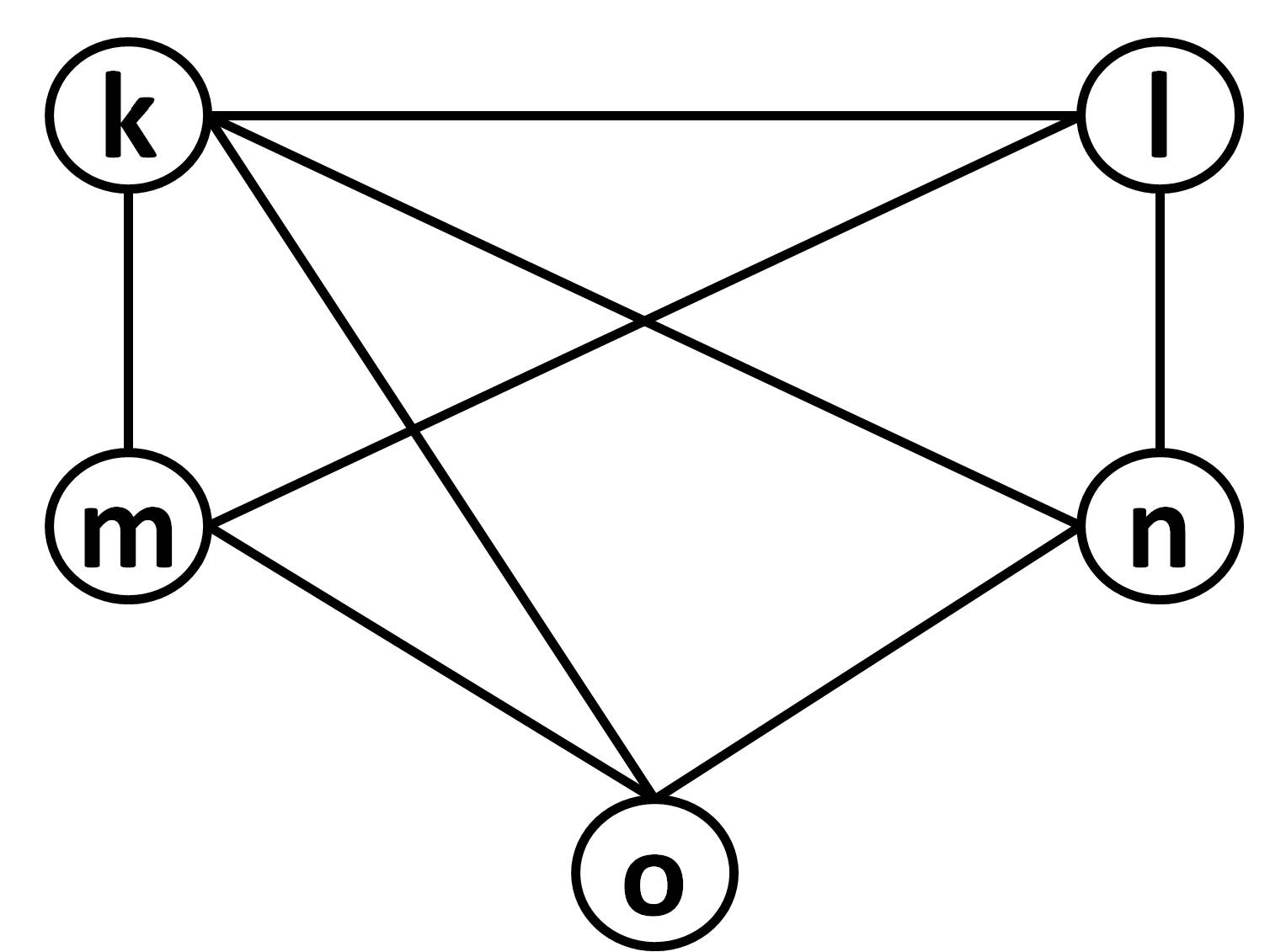} 
\label{fig:delete-network-3}
}
\caption{Stable Network $\mathfrak{g}$ on 15 agents, consisting of 3 components $\mathfrak{g}(\kappa_{1})$, $\mathfrak{g}(\kappa_{2})$ and $\mathfrak{g}(\kappa_{3})$}
\label{fig:3odd}
\end{figure*}
\end{example}

\begin{clm}\label{claim2}
Suppose $\mathfrak{g}$ is a symmetric social storage network. 
If $\hat{\eta}=1$, and if $\mathfrak{g}$ has evolved from the null network, then $\mathfrak{g}$ is bilaterally stable if and only if $\mathfrak{g}$ consists of a set of $\frac{N-1}{2}$ connected pairs of agents plus one isolated agent if $N$ is odd, and a set of $\frac{N}{2}$ connected pairs of agents if $N$ is even.
\end{clm}
\begin{proof}
\textcolor{mycolor}{As $\mathfrak{g}$ has evolved from the null network and as $\hat{\eta}=1$, no agent has two or more neighbours. Hence, if  $N$ is even, $\mathfrak{g}$ consists of $\frac{N}{2}$ connected pairs of agents. 
 Similarly, if $N$ is odd, $\mathfrak{g}$ consists of one isolated agent and the remaining $N-1$ agents connect in pairs.}\qed
\end{proof}

\begin{remark}\label{remark-from-ccomplete-star}
In Claim \ref{claim2}, if $\mathfrak{g}$ has evolved from the complete network (by mutual deletion of links), then networks consisting of star components are also bilaterally stable as per Definitions \ref{def:stability}, \ref{def:stability-storage-constraints}, and \ref{def:stability-storage-budget-constraints}.  
\end{remark}

\begin{remark}
In Claim \ref{claim2}, if $\mathfrak{g}$ is a given network, then in addition to the star components as discussed in Remark \ref{remark-from-ccomplete-star}, 
$\mathfrak{g}$ may also consist of (at most) one isolated agent and continue to be bilaterally stable.  
\end{remark}

It is interesting to note that in any star network, given that $\hat{\eta}=1$, though the universal agent has incentive to delete a link (or links), no other (pendant) agent will consent to deletion. 
However, if we start from the null network, we have the following observation.\\

\begin{clm}\label{claim3}
Suppose $\mathfrak{g}$ has evolved from the null network. Then, if  $\mathfrak{g}$ is bilaterally stable, $\mathfrak{g}$ can never contain a star network as component.
\end{clm}

\begin{proof}

\textcolor{mycolor}{If $\hat{\eta}=1$, the result follows from Claim \ref{claim2}}.

\textcolor{mycolor}{Suppose $\hat{\eta}>1$. If possible, let $\mathfrak{g}$ be a star network. It is easy to see that all pendant agents have incentives to add (at least) one more link implying that $\mathfrak{g}$ is not bilaterally stable, a contradiction.}
\qed
\end{proof}

\subsubsection{Unique Stable Networks}\label{subsubsec:unique}

In the previous subsections, we have seen results on the existence of a bilaterally stable social storage network. In this subsection, we look at conditions under which a unique bilaterally stable social storage network exists. Whenever a unique bilaterally stable network exists, the agents themselves endogenously form this network. Any independent observer or regulator knows precisely which network would form (or has formed).

\begin{clm}
If $N=\hat{\eta}+1$ or $\hat{\eta} \geq N$, then there exists a unique symmetric social storage network $\mathfrak{g}$ that is bilaterally stable, namely the complete network on $N$ agents.
\end{clm}
\begin{proof}
\textcolor{mycolor}{In both cases (that is, $N=\hat{\eta}+1$ or $\hat{\eta} \geq N$), the complete network is the one which maximises the utility of each agent. That is, no agent has an incentive to delete any existing link and, clearly, no agent can add any more links.} \qed
\end{proof}

\begin{clm}
If $N>\hat{\eta}+1$, then there are always two or more different (with respect to degree sequence\footnote{Two networks are {\em{different with respect to degree sequence}} if the sorted sequence of degrees (neighborhood sizes) in one is different from that of the other. Note that, both sequences are sorted in the ascending order (or both in the descending order).}) bilaterally stable networks.
\end{clm}
\begin{proof}
\textcolor{mycolor}{If $N=\hat{\eta}+2$, the following  stable networks are possible, which are different with respect to degree sequence.}
\textcolor{mycolor}{The first, where $N-1$ agents form a clique and the other agent is isolated. The second network is as follows. If $\hat{\eta}$ and, hence, $N$ are even, the connected regular network with $N$ agents, where each agent has a neighborhood size of $\hat{\eta}$ is bilaterally stable. If $\hat{\eta}$ and, hence, $N$ are odd, the connected network with $N$ agents, where $N-1$ agents have a neighborhood size of $\hat{\eta}$ and the other agent has a neighborhood size of $\hat{\eta} - 1$, is bilaterally stable. (If $\hat{\eta}$ and $N$ are odd, the connected network where $N-1$ agents have a neighborhood size of $\hat{\eta}$ and the other agent has a neighborhood size of $\hat{\eta} + 1$, is a third bilaterally stable network).} \qed
\end{proof}
\begin{example}
\begin{figure}[!h]
\centering
\subfigure[$\mathfrak{n}_{1}$ ]
{
\includegraphics[scale=0.15,width=2cm,height=1.8cm]{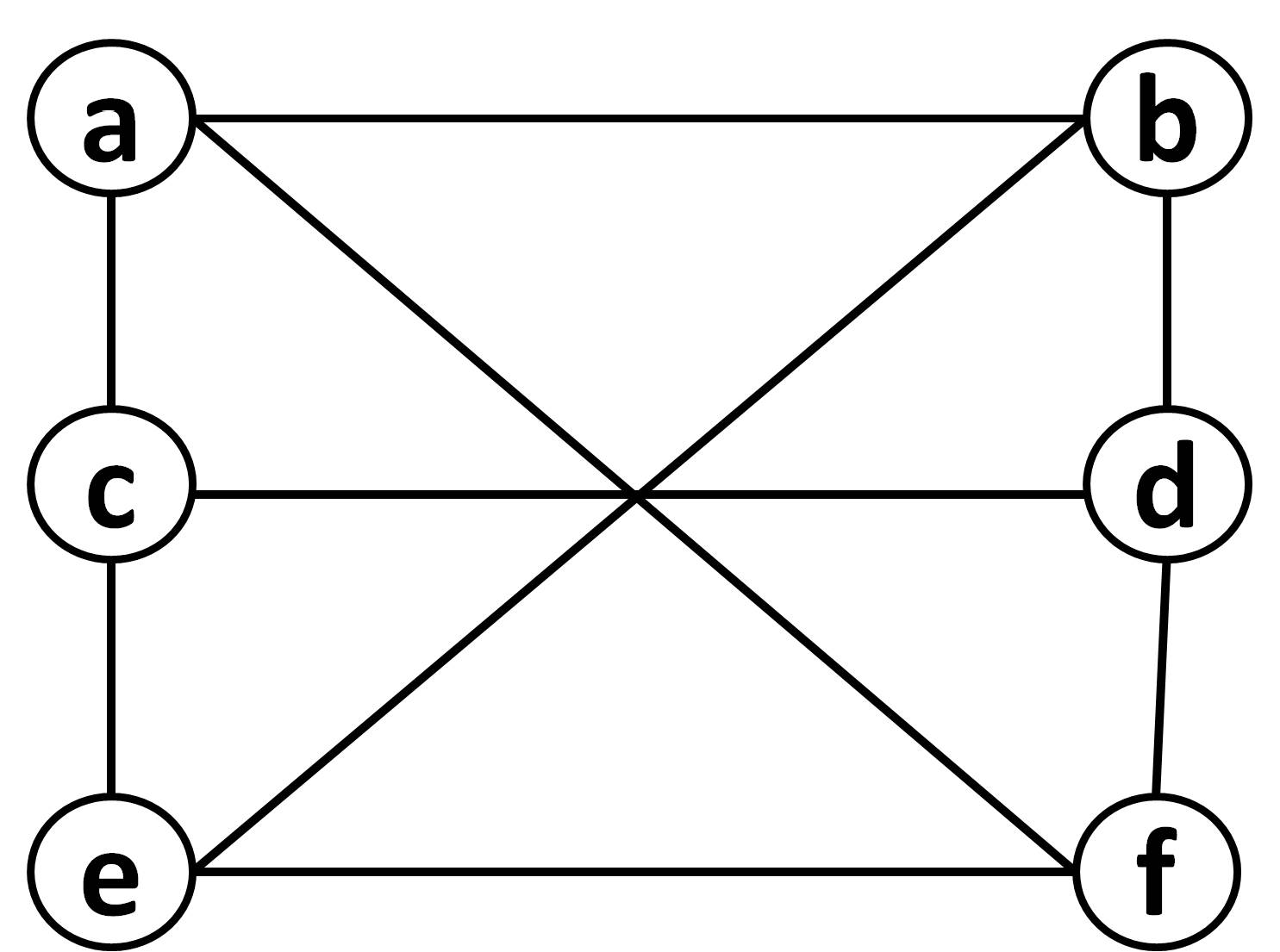} 
\label{fig:n1}
}
\quad
\subfigure[ $\mathfrak{n}_{2}$]
{
\includegraphics[scale=0.15,width=2cm,height=1.8cm]{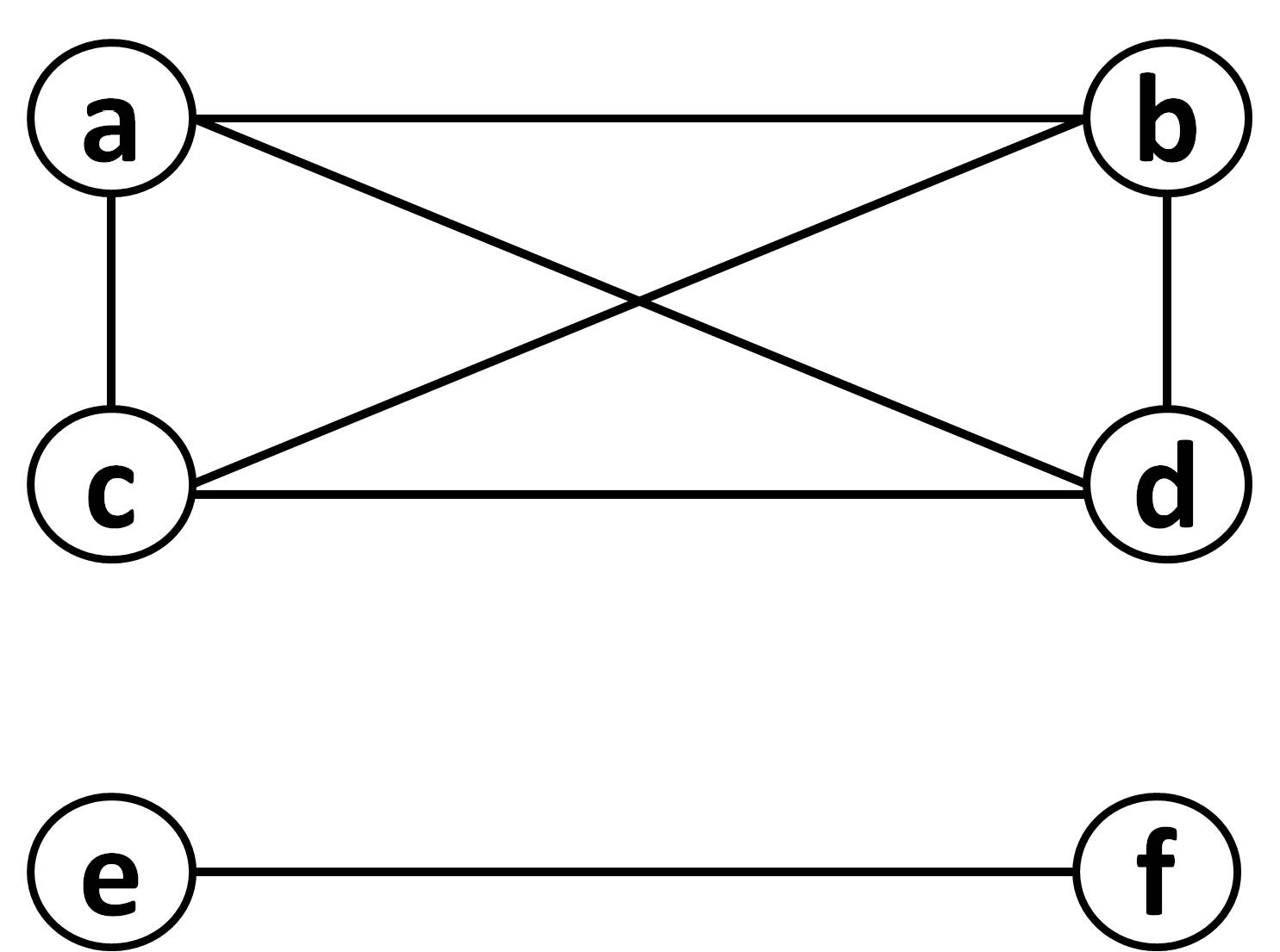} 
\label{fig:n2}
}
\quad
\subfigure[ $\mathfrak{n}_{3}$]
{
\includegraphics[scale=0.15,width=2cm,height=1.8cm]{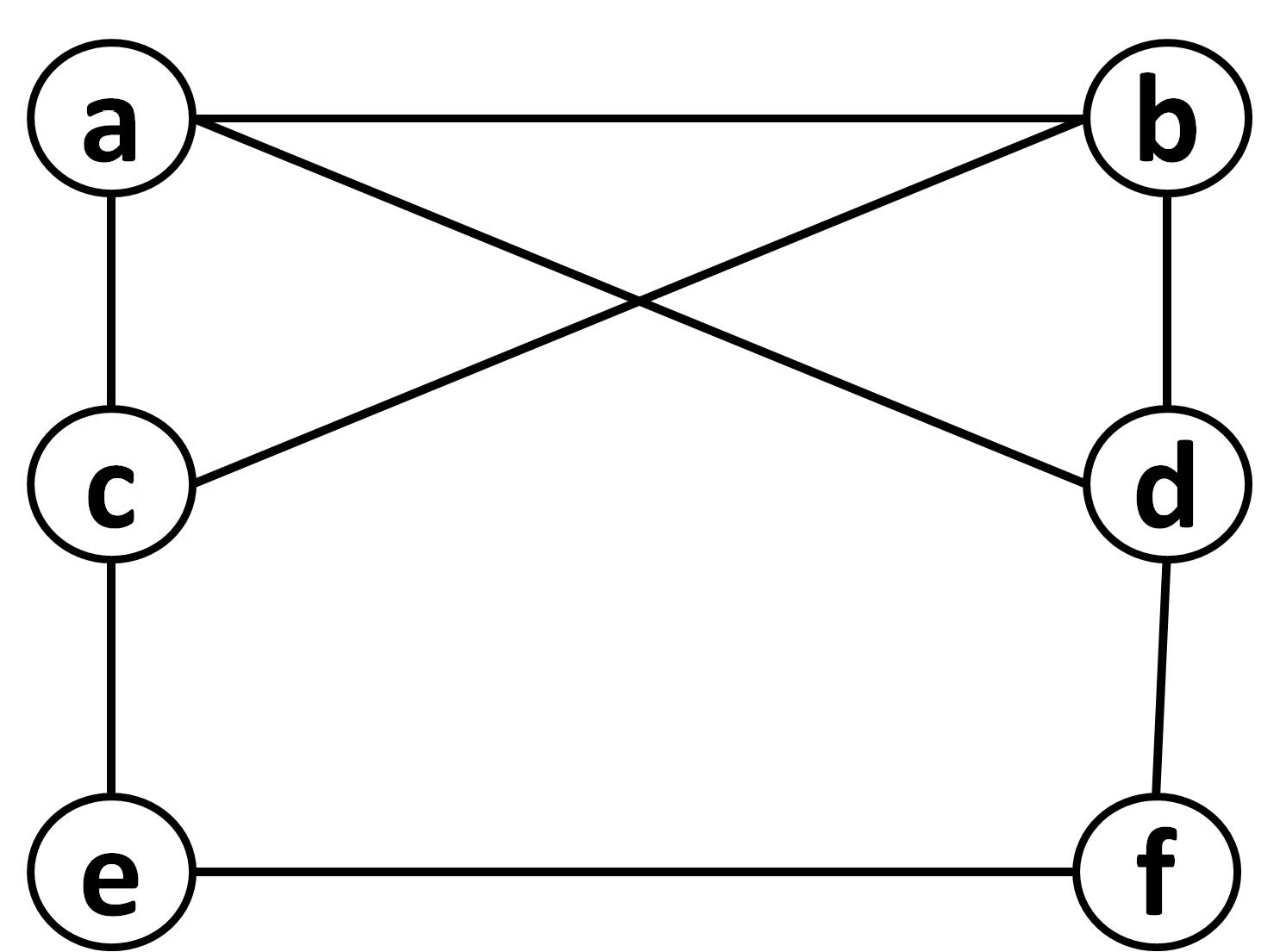} 
\label{fig:n3}
}
\quad
\subfigure[ $\mathfrak{n}_{4}$]
{
\includegraphics[scale=0.15,width=2cm,height=1.8cm]{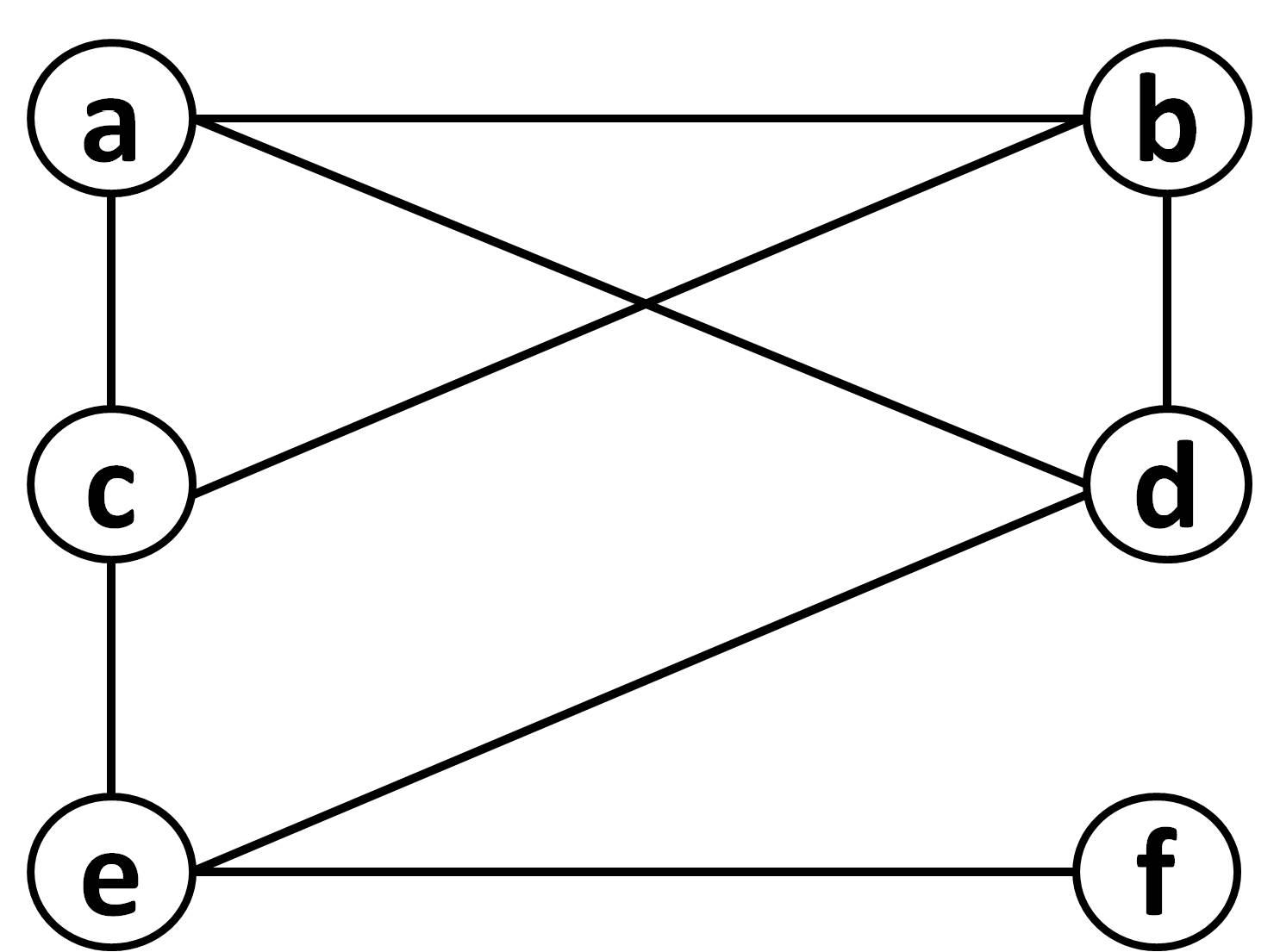} 
\label{fig:n4}
}
\caption{Bilaterally Stable Networks with $N=6$ agents, $\hat{\eta}=3$.}
\label{fig:network-s}
\end{figure}
Let $\hat{\eta}=3$ and $N=6$. Then, there are four networks $\mathfrak{n_{1}}$ (see Fig. \ref{fig:n1}), $\mathfrak{n_{2}}$ (see Fig. \ref{fig:n2}), $\mathfrak{n_{3}}$ (see Fig. \ref{fig:n3}) and $\mathfrak{n_{4}}$ (see Fig. \ref{fig:n4}) which are bilaterally stable. \\
\end{example}

If we look at specific protocols of network formation, then we get further uniqueness results. For example, in Claim \ref{claim2}, starting from the null network (or any network where no agent has more than $1$ neighbor), the resulting bilaterally stable network is unique up to isomorphism.

\subsection{Efficient and Contented Social Storage Networks}\label{subsec:welfare}

In this subsection, we look at efficient social storage networks and contented social storage networks. As discussed earlier, an observer who observes or monitors or regulates the network may externally perturb the system so as to reach an efficient or a contented network. 


We have seen in Section \ref{sec:stability-point} that there exists a unique stability point, $\hat{\eta}$, (for each network type, under the given framework) such that, no agent gains by adding more neighbors than $\hat{\eta}$, and severing existing relationships resulting in a neighborhood size of less than $\hat{\eta}$. An efficient social storage network is, hence, one in which maximum possible number of agents have $\hat{\eta}$ neighbors. 

\begin{remark}
An efficient social storage network is bilaterally stable. 
\end{remark} 

We, now, discuss an example to highlight the fact that not all stable networks are efficient. 

\begin{figure}[!h]
\centering
\subfigure[$\mathfrak{g}_{1}$]
{
\includegraphics[scale=0.14]{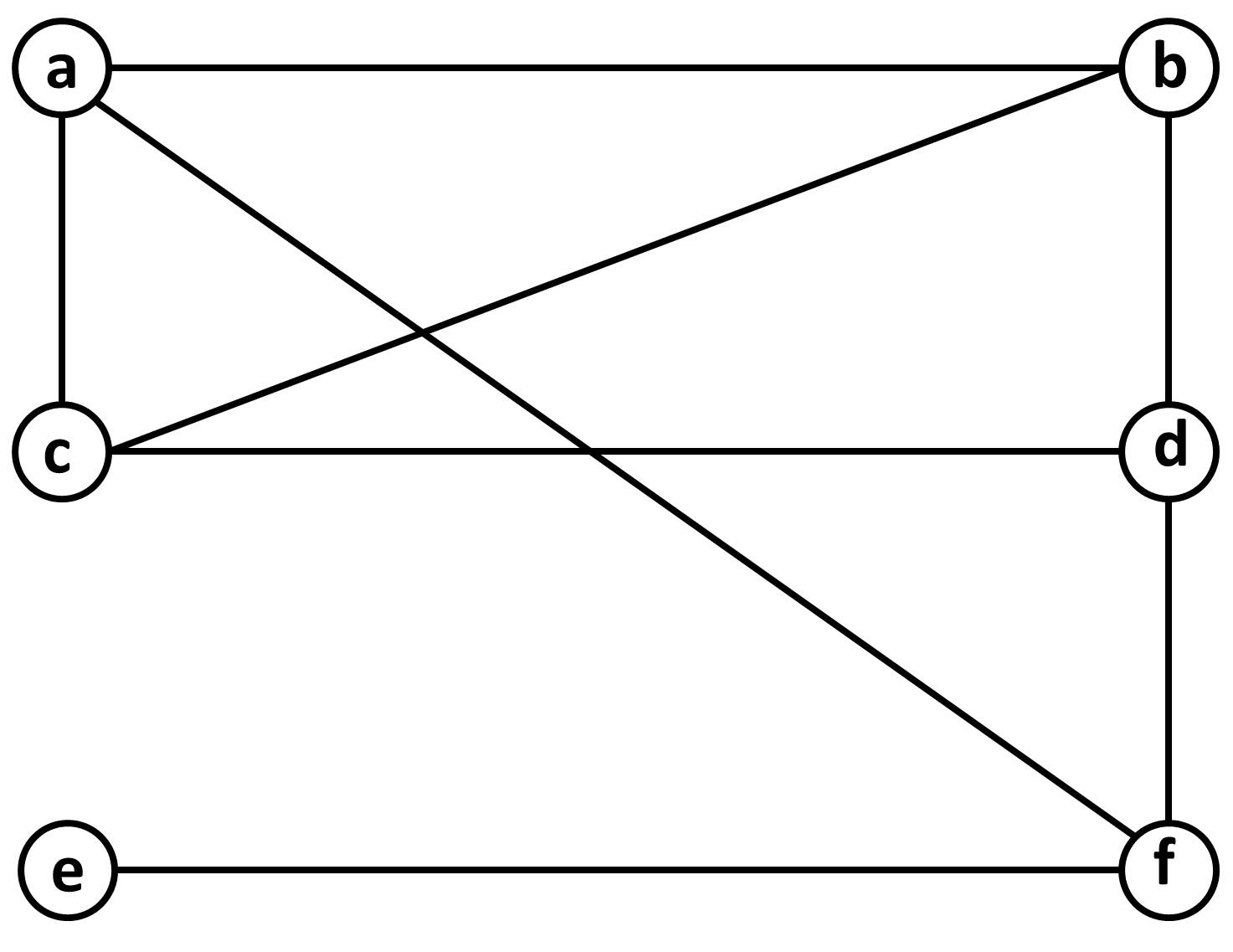}
\label{fig:struct-welfare-g1}
}
\quad
\subfigure[$\mathfrak{g}_{2}$]
{
\includegraphics[scale=0.14]{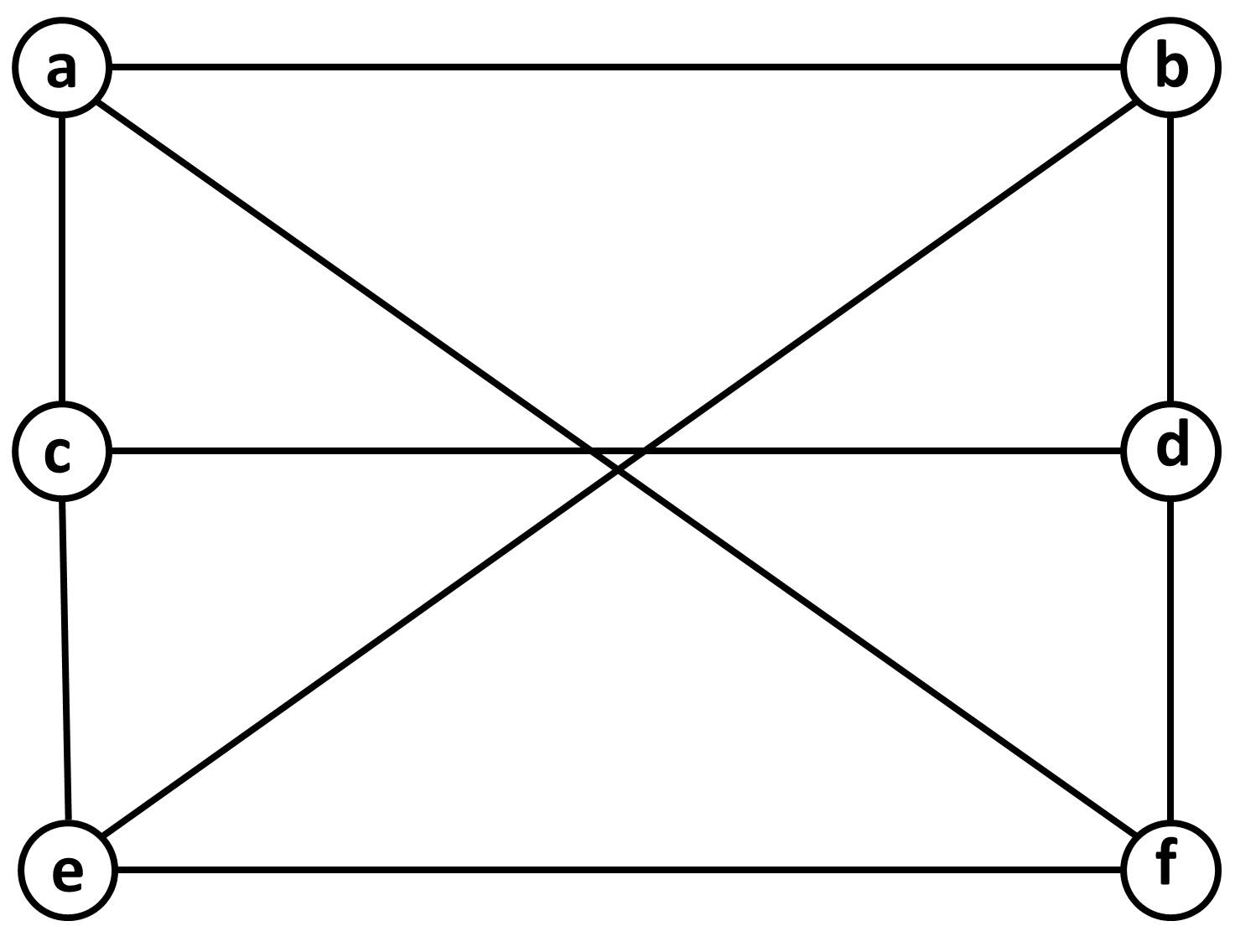}
\label{fig:struct-welfare-g2}
}
\quad
\subfigure[$\mathfrak{g}_{3}$]
{
\includegraphics[scale=0.14]{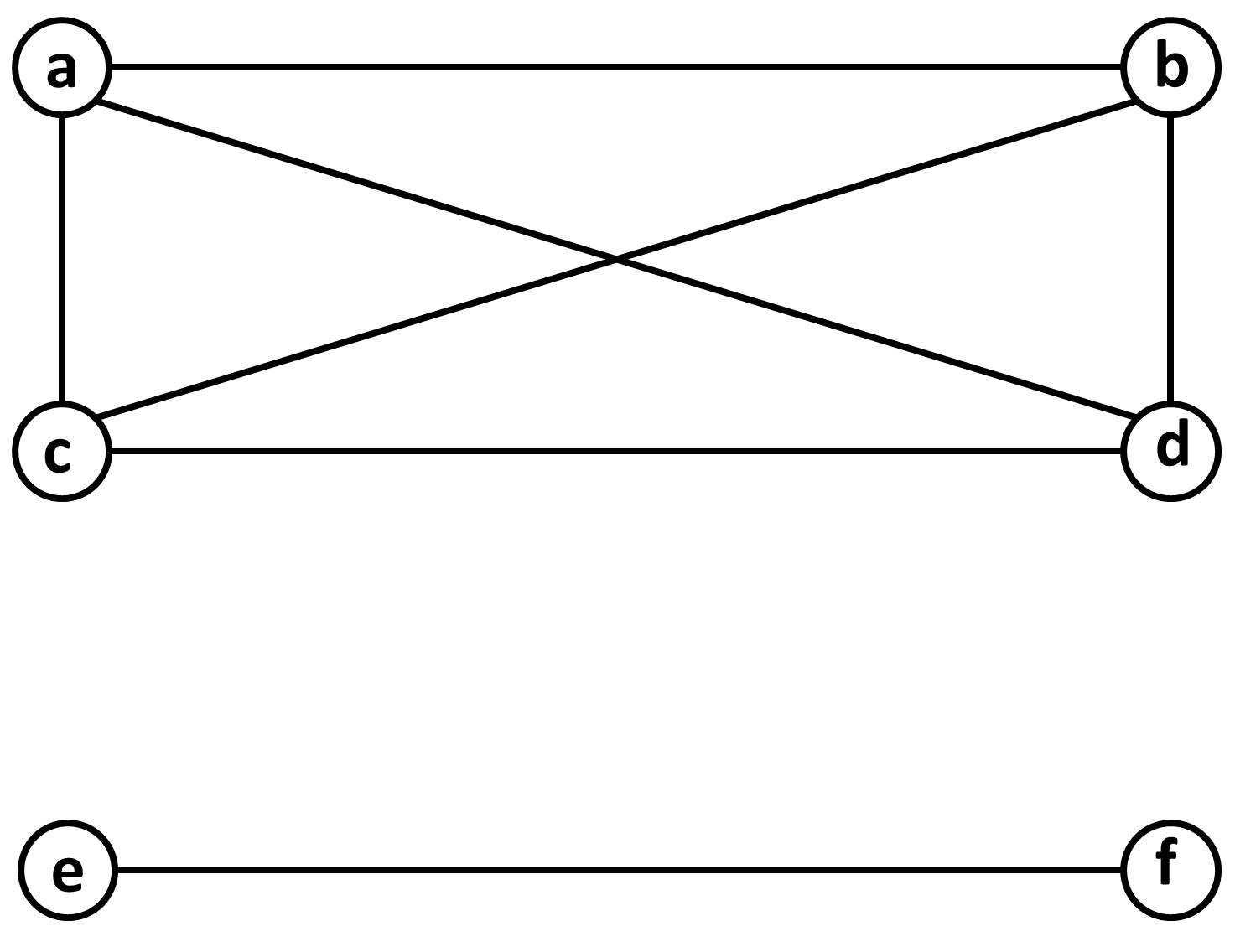}
\label{fig:struct-welfare-g3}
}
\caption{Network Structure and Social Welfare}
\label{fig:network-structure-welfare}
\end{figure}
\begin{example}\label{exmp:social-welfare}
Suppose there are six agents, $a, b, c, d, e$ and $f$, in a social storage network with stability point $\hat{\eta}=3$. Assume that, starting from the null network, these agents add links (that is, build mutual data backup partnerships). Different network structures  may emerge, for example Fig. \ref{fig:struct-welfare-g1}, Fig. \ref{fig:struct-welfare-g2}, and Fig. \ref{fig:struct-welfare-g3}).
\end{example}
In network $\mathfrak{g}_{1}$ (see Fig. \ref{fig:struct-welfare-g1}), agent $e’s$ expected value of data backup is less than that of the rest of the agents. In $\mathfrak{g}_{2}$ (see Fig. \ref{fig:struct-welfare-g2}), all agents achieve the same (and maximum) expected value of data backup, and in $\mathfrak{g}_{3}$ (see Fig. \ref{fig:struct-welfare-g3}), agents $a, b, c$, and $d$ achieve higher 
expected value of data backup than agents $e$ and $f$. $\mathfrak{g}_{2}$ is efficient, whereas $\mathfrak{g}_{1}$ and $\mathfrak{g}_{3}$ are not (though they are bilaterally stable). 

\vskip 1em

We, now, discuss contented networks. 
%

\begin{remark}
A contented social storage network is bilaterally stable. 
\end{remark}

It is easy to see that not all stable networks are contented. In Example \ref{exmp:achieve-stability}, though both $\mathfrak{g}$ and $\mathfrak{s}$ are stable, neither of these networks are contented. 
Consider $\mathfrak{g}$. An independent observer could just add a storage device, $p$, to the network, which leads to a contented network as explained below. 
This storage device acts as a dummy agent, not trying to maximise its utility, and always agreeing to add or delete any link with any agent. 
Hence, for contented networks, we do not consider any dummy agent as a part of the network. 
In $\mathfrak{g}$, agent $e$ has not achieved the maximum possible utility (as $\hat\eta = 3$, but $e$ has 2 neighbors) while all other agents have. By allowing $e$ to store a copy of its data on storage device $p$, $e$ also obtains the maximum possible utility. This network is, now, a contented network (where the utility of the dummy agent $p$ is not considered). This is, in fact, a hybrid model --- hybrid between a centralised storage system and a decentralised one.

Next, we relate contented networks and efficient networks. 

\begin{proposition}\label{prop:content-efficient}
Let $\mathfrak{g}$ be a symmetric social storage network with $N$ agents and stability point $\hat\eta$. Then: 
\begin{enumerate}
\item If at least one of $N$ and $\hat\eta$ is/ are even, then, $\mathfrak{g}$ is efficient if and only if $\mathfrak{g}$ is contented. 
\item Suppose $N$ and $\hat\eta$ are odd. Then, an efficient network does exist but there does not exist any contented network. 
\end{enumerate}
\end{proposition}
\begin{proof}
{\it 1} follows from Proposition \ref{lem:g-is-even-stable}, since, if at least one of $N$ and $\hat\eta$ is/ are even, then, $\mathfrak{g}$ is efficient if and only if $\mathfrak{g}$ is $\hat\eta$-regular. \\

{\it 2} follows from Proposition \ref{lem:n-1-agent}.  \qed
\end{proof}

\vskip 1em
\begin{remark}
Not all efficient networks are contented. 
\end{remark}

In Example \ref{exmp:achieve-stability}, neither $\mathfrak{g}$ nor $\mathfrak{s}$ are contented. However, (at least) one of them is efficient. The following Propositions help identify which of them is/ are efficient, under the MO- as well as SO-Frameworks, for SVN, SRN and SV-SRN networks, as the case may be (Refer Table \ref{table:Different-Networks}). Note that any stable network in which maximum possible number of agents have $\hat{\eta}$ neighbors is not necessarily efficient, as per Definitions \ref{def:welfare}, \ref{def:welfare-storage-constraint}, and \ref{def:welfare-storage-budget-constraints}. 

\begin{proposition}\label{mo-iff-efficient}
Let $\mathfrak{g}$ be an SVN or SV-SRN under the MO-Framework, with $N$ agents and stability point $\hat\eta$. Suppose both $N$ and $\hat\eta$ are odd. 
Then $\mathfrak{g}$ is efficient if and only if $\mathfrak{g}$ has $N-1$ agents with neighborhood size $\hat\eta$ and one of the following holds: 
{
\begin{enumerate}
\item $c<\frac{\beta \lambda^{^{\hat\eta}}}{2}(\frac{1}{\lambda}-\lambda)$ and $\mathfrak{g}$ has one agent with neighborhood size $\hat\eta + 1$.
\item $c>\frac{\beta \lambda^{^{\hat\eta}}}{2}(\frac{1}{\lambda}-\lambda)$ and $\mathfrak{g}$ has one agent with neighborhood size $\hat\eta - 1$.
\item $c=\frac{\beta \lambda^{^{\hat\eta}}}{2}(\frac{1}{\lambda}-\lambda)$ and $\mathfrak{g}$ has one agent with neighborhood size either $\hat\eta + 1$ or $\hat\eta - 1$.
\end{enumerate}
}
\end{proposition}
\begin{proof}

For each $i \in \mathbf{A}$, its utility is   
$u_{i}(\mathfrak{g})=\beta_{i}(1-\lambda^{\eta_{i}(\mathfrak{g})})-c \eta_{i}(\mathfrak{g})$. (Refer Equation \ref{eq:utility-mo-frame}). 
As the network is SVN or SV-SRN, $\beta_{i} = \beta$, for all $i$. 
\begin{equation*}\label{efficient-mo}
\underset{\eta_{i}(\mathfrak{g})} {max\ } u_{i}(\mathfrak{g}) = \underset{\eta_{i}(\mathfrak{g})} {max\ }(\beta(1-\lambda^{\eta_{i}(\mathfrak{g})})-c \eta_{i}(\mathfrak{g})) 
= \beta(1-\lambda^{\hat\eta})-c \hat\eta
\end{equation*}

Let $\mathfrak{g_1}$ be the network where 
$N-1$ agents have $\hat\eta$ neighbors and the other agent has $\hat\eta - 1$ neighbors. 
Let $\mathfrak{g_2}$ be  the network where  
$N-1$ agents have $\hat\eta$ neighbors and the other agent has $\hat\eta + 1$ neighbors. 

\vskip 1em

$u_i(\mathfrak{g_1}) = \beta(1-\lambda^{(\hat\eta - 1}))-c (\hat\eta - 1)$ and $u_i(\mathfrak{g_2}) = \beta(1-\lambda^{(\hat\eta + 1}))-c (\hat\eta + 1)$. 

\vskip 1em

From Proposition \ref{lem:n-1-agent} and Definition \ref{def:welfare}, it is easy to see that either $\mathfrak{g_1}$ or $\mathfrak{g_2}$ (or both) is (are) efficient. 
That is,  
$max (\sum\limits_{i}u_i(\mathfrak{g}))$ is either $u_i(\mathfrak{g_1})$  or $u_i(\mathfrak{g_2})$. If $u_i(\mathfrak{g_1}) < u_i(\mathfrak{g_2})$, we get result $1$,  
if $u_i(\mathfrak{g_1}) > u_i(\mathfrak{g_2})$, we get result $2$ and if $u_i(\mathfrak{g_1}) = u_i(\mathfrak{g_2})$, both $\mathfrak{g_1}$ and $\mathfrak{g_2}$ are efficient, leading to result $3$. \qed


\end{proof}

\begin{proposition}
An SVN under the SO-Framework is efficient if and only if it is a complete network. 
\end{proposition}
\begin{proof}
Follows from Theorem \ref{lemma:stability-point-susb-network}. \qed
\end{proof}

\begin{proposition}\label{prop:SRN-SO-Contentedness-Framework}
For an SRN under the SO-Framework, $s/d$ and $b/c$ act as constraints for the maximum number of links possible. 

Let $\mathfrak{g}$ be an SRN under the SO-Framework, with $N$ agents and stability point $\hat\eta$. Then: 
\begin{enumerate}
\item If at least one of $N$ and $\hat\eta$ is/ are even, then $\mathfrak{g}$ is efficient if and only if $\mathfrak{g}$ is $\hat\eta$-regular. 
\item If both $N$ and $\hat\eta$ are odd, then $\mathfrak{g}$ is efficient if and only if $\mathfrak{g}$ has
$N-1$ agents with $\hat\eta$ neighbors and the other agent with $\hat\eta - 1$ neighbors. 
\end{enumerate}
\end{proposition}
\begin{proof}
From Theorem \ref{lemma:stability-point-sb-network}, $\hat\eta = min\{\frac{s}{d}, \frac{b}{c}\}$. 

As the budget $b$ and the storage space available $s$ act as constraints, no agent can have more than $\hat\eta$ neighbors. Therefore, part $1$ follows from Proposition \ref{lem:n-1-agent}. 
(We do not have the possibility of one agent having $\hat\eta+1$ neighbors as we had in Proposition \ref{mo-iff-efficient}). 

Part $2$ follows from Proposition \ref{lem:g-is-even-stable}. \qed
\end{proof}

\section{Related Work}\label{sec:related-future}
\myworries{As briefly mentioned in the introduction, our work is most closely related to P2P systems. Especially, our strategic network formation game has some similarities with peer selection for data placement \citep{Rzadca-Game-Free-Riding, Rzadca-Extended, Toka-Partner-Selection} and topology formation \citep{Moscibroda-2006-topology-formation, stefan-network-creation-game} in P2P systems.  To start off, in P2P nomenclature, 
virtual (i.e., logical) topologies (or structures) are built by peers (in-general computers or software modules) on top of physical networks (e.g., Internet).} 


\myworries{\cite{Rzadca-Game-Free-Riding, Rzadca-Extended} have studied data placement in a strategic interaction between peers to maximise data availability. Here, peers are involved in a reciprocal replication contract (a pair of agents replicates each others' data to increase data availability). They show that agents prefer to form contracts with only those who have similar availability. This behaviour of peers makes the system inefficient. 
However, by setting cooperation rules and providing incentives to peers, data availability can be increased along with the increase in the efficiency of the system. We can take inspiration from these ideas of cooperation rules and incentives to design a more practical social storage system.}

\myworries{\cite{Toka-Partner-Selection} studied data placement in a different strategic setting than above. Here, peers selfishly select partners based upon their profiles. The profile of each peer, which includes the online availability, the bandwidth, and global preferences, is considered along with the utility function so that the data storing costs are minimised. 
Authors have shown here that there exists at least one pairwise stable matching and it can be found in polynomial time. In our study, we do not consider agents availability. This is based upon the assumption that the out of band communication \citep{f2fstorage} is possible between them. However, this can be further analysed.}

\myworries{Finally, we look at topology formation in P2P systems. P2P topologies are a mirror image of social connections in our case.
\cite{Moscibroda-2006-topology-formation, stefan-network-creation-game} have proposed a locality game (inspired by the network creation game proposed by \cite{Fabrikant-NCG-2003}) to study the impact of selfish peers on P2P topologies. In this setting, selfish peers select their partners in such way that the stretch (i.e., the look up performance in terms of latencies) could be minimised. Their three main results are as follows: the topologies build by selfish agents are worse compared to the topologies build by agents in collaboration;  the topologies constructed by selfish agents are never stable (i.e., there is always a change in the topology); 
and determining a pure Nash equilibrium is NP-complete here.  This aspect of selfish agents is part of our future work and is discussed in detail in the next section. However, as motivated in the introduction and in the background sections, for us the solution concept of Nash equilibrium is not useful and we use bilateral stability instead.}

\myworries{When looking at P2P systems more closer to our social storage systems, P2P social networking is one such area \citep{Buchegger-P2P-Social-Network-1, Buchegger-P2P-OSN}. Topology formation is one of the concerns here \citep{Buchegger-P2P-OSN}. We believe that our solution concept of bilateral stability has its theoretical consequences in determining which bilaterally stable topology emerges in P2P social networking.}


\section{Conclusion and Future Work}\label{sec:conclusion}
In this paper, we have expanded on two untouched aspects of social storage systems, namely, endogenous network formation and bilateral stability of such networks. We have formalised social storage networks as a network formation game where each agent tries to maximise its utility.
We considered two frameworks for utility of agents in the network. We modified the pairwise stability definition of \cite{jackson} to include mutual consent for link deletion too (as required for social storage networks), and also to include storage and budget constraints. 

After defining bilateral stability as a modification of pairwise stability, we analysed bilateral stability of symmetric social storage networks. Our stability analysis involved restudying conditions of stability under the new definition of pairwise stability (that is, bilateral stability), derivation of a unique stability point (which is a neighborhood size where no agent has any incentive to add or delete a link), and some necessary and sufficient conditions for symmetric social storage networks to be bilaterally stable. We also showed that ideally all agents in a network want to achieve their stability point but a network can be bilaterally stable even when this stability point is not reached for one agent. 

\myworries{Further, we discussed which bilaterally stable networks would evolve. We also discuss why just studying stability is not enough and one has to look at efficiency and contentment of the network. Efficiency is the case when the sum of utilities of all agents is maximised, and contentment is when the individual utility of every agent is maximised. We relate these three properties of the network with one another. We also give conditions on the number of agents and stability point (besides other constraints) to achieve bilaterally stable, efficient, and contented networks.}

\myworries{Next, we discuss some future directions. We first discuss model specific future work and then solution concept specific. We assumed that the cost to maintain a link is shared equally among the agents on either side of the link. Looking at asymmetric cost sharing, for example centrally-sponsored star networks, is one of the direction for future work.} 

\textcolor{mycolor}{For the MO-framework, we use a convex combination of our two objective functions (maximizing data reliability and minimizing the total cost of the link), and this is no longer a case of Multi-Objective (MO) optimization. Since the solution of the convexly combined problem may not always be the solution of the original MO problem, we plan to look at finding a Pareto frontier as part of future work (path followed by most MO algorithms).}

\myworries{In our current work, we have not focused on the heterogeneous behavior of agents in social storage settings. Although incorporating complex and heterogeneous behavior of agents into the model is closer to real world scenarios, this would make it difficult to deal with the model and as well as predict its outcome. \cite{stefan-social-range-matrix} propose a social range matrix, which is a novel approach to deal with heterogeneous behavior of agents in the network. In particular, social range matrices capture three scenarios: anarchy, monarchy and coalitions. In anarchy, each agent is selfish. In monarchy, agents only care about one agent in the network. In the coalitions scenario, agents support each other within the same coalition but act selfishly or maliciously towards agents in other coalitions. In this work, they propose a network creation game for capturing the effect of the social range matrix, and further explore how this matrix affects equilibria in a network game. 
Investigating the applicability of the social range matrix for the frameworks \ref{subsub:framework1} and \ref{subsub:framework2} and bilateral stability is part of our future work.}

\myworries{In all our discussions,
we have assumed that any pair of agents can potentially form a link. In scenarios where agents do not necessarily trust all agents in the network, our results on bilateral stability extend to every clique (of mutually trusting agents) in the network. If not all agents trust each other, we may use an extension of the Hall's marriage theorem \citep{hall1935representatives} to aid independent observers determine whether it is possible to form an efficient network or not.} 

\myworries{Coming to the solution concept, if we had used the concept of {\em{Pairwise Nash Stability}} as defined by \cite{goyal-pairwise-nash} and had applied the mutual consent requirement for deletion too, we would have the same results we have obtained in this paper. This is because the mutual consent requirement for addition and deletion overrides the requirement for Nash equilibrium. We are currently working on modifying the definition of Pairwise Nash Stability to multiple other scenarios.
Looking at strong and coalition-proof Nash equilibria \citep{duttamutuswami}, strong pairwise stability \citep{jackson2005strongly}, and farsighted equilibrium \citep{Dutta-Farsighted}, are also future research directions.} 

\myworries{In this paper, we have discussed about network efficiencies but not looked at its contrapositive. That is, the inefficiencies in the network. The price of anarchy is an interesting measure to analyse the extent to which a network is inefficient \citep{Papadimitriou-POA, demaine-price-of-anarchy}. By definition, this means ratio of the worst sum of the utilities of agents in an equilibrium network to the best sum of the utilities. In this paper, we have bilaterally stable networks in place of equilibrium networks. In our case, efficient networks are the ones with the best sum of utilities. We also plan to analyze stability and efficiency of social storage networks by considering pairwise stability as the solution concept, thereby dropping the requirement of mutual consent for deletion.} 

\myworries{Our utility function depends on the cost incurred by an agent to maintain a link, the worth (value) of data, the disk failure rate, and the neighborhood size of an agent. Out of these, the first three are constants (same for all agents) while the neighborhood size varies. In our case, we can find the best sum of utility (using Proposition \ref{prop:SRN-SO-Contentedness-Framework}). However, knowing the worst sum of utility is non-trivial. The neighborhood size of every agent that would give us the worse sum of utility is challenging.}




\end{document}